\documentclass[a4,50pt]{amsart}%
\usepackage{amssymb}
\usepackage{pdflscape}
\usepackage{amsmath}
\usepackage{mathrsfs}
\usepackage{enumerate}
\usepackage{amsfonts}
\usepackage{graphicx}%
\usepackage[all]{xy}
\usepackage[colorinlistoftodos]{todonotes}
\usepackage[colorlinks=true, allcolors=purple]{hyperref}
\newtheorem{Theorem}{Theorem}

\theoremstyle{plain}

\newtheorem{Corollary}{Corollary}

\newtheorem{Definition}{Definition}

\newtheorem{Lemma}{Lemma}

\newtheorem{Remark}{Remark}

\numberwithin{equation}{section}
\begin{document}

\title[filtering]{On Observation and The Completion of Quantum Mechanics}
\author{M.~F.~Brown}
\date{October 2023. Typo corrections and revised conclusion, January 2024}

\begin{abstract}
We start with a discussion of the use of mathematics to model the real world then justify the role of Hilbert space formalism for such modelling in the general context of quantum logic. Following this, the incompleteness of the Schr\"odinger equation is discussed as well as the incompleteness of von Neumann's measurement approach \cite{vN}. Subsequently, it is shown that quantum mechanics is indeed completed by the addition of an observer, however the observer is not described in the Hamiltonian formalism but \emph{necessarily} by the quantum stochastic formalism discovered in \cite{HP}. Consequently, the complete theory of quantum mechanics appears to be the Quantum Filtering Theory \cite{ND,NLF}. Finally, it is shown how Schr\"odinger's cat may be understood as a quantum filter, providing an intuitively realistic model and an insight into how  quantum filtering works.
\end{abstract}


\maketitle
\tableofcontents


\section{Introduction}\label{1}
Hilbert space, or equivalently $C^\ast$-algebra, is the mathematical framework of quantum mechanics but do we need to justify its use in modelling the world, or parts thereof? After all, it is rather unusual for anyone to question, or justify, the use of vector calculus as the mathematical framework of electromagnetism. Of course vector calculus makes sense as far as our abilities to `imagine', `conceptualise' and `comprehend' are concerned. However, it should be exposed that this is an attempt to model a certain class of observations using a specific mathematical structure, it is not `reality'. Instead, it is a mathematical framework that helps us make sense of our observations. If we believe that a model \emph{is} reality then we run the risk of constraining our subsequent understanding of the world, becoming stuck in particular systems of mathematics and logic.
That said, we do need logical structures if we are to attempt to understand the world around us, because our understanding is determined by our ability to predict and thus engineer.

Vector calculus is a good logic for imagining a traditional 3-dimensional space of evolving trajectories, naturally giving rise to space-time. So good in fact, that most people would not question it as truth. But what about $C^\ast$-algebra? In quantum mechanics we are usually more interested in $W^\ast$-algebras (von Neumann algebras). A von Neumann algebra $\mathcal{A}$ is a special kind of $C^\ast$-algebras (it is equal to its bicommutant) and as such, may be represented by an operator algebra on a Hilbert space. The building blocks of such algebras are the projection operators, which are a fundamental ingredient of Quantum Mechanics, and von Neumann algebras can be classified according to the nature of their projectors. So, when it comes to describing the world using Hilbert space the important thing to understand is how the mathematical logic of projectors captures the nature of the world around us.

On the one hand, we can see geometry around us and watch things moving along trajectories - a vector calculus outlook. On the other hand we can take less obvious (more abstract) approaches to understanding what we see. Notice that one of the key ideas, that has been implicit up to now, is that when we are trying to describe {the} world it is really just our \emph{observation of} the world that we are describing. That is, the best we can do is use logic to understand our (own or collective) experience of the world. Physics attempts to resolve this by replacing ourselves, as observers, with an apparatus of some kind, an `objective observer'. This will achieve a certain level of objectivity in an otherwise subjective world, but this objectivity is somewhat determined by us. This is because the apparatus that we build is conceptually driven by the logic that we have chosen to apply to the world. For example, if we only understand the world in terms of vector calculus, then we would engineer an apparatus to measure some property related to this vector calculus model. It is usually when our apparatus produces outcomes that we don't understand using our current model, that we start wondering if perhaps our model has not as good a likeness to the universe as we thought. Alternatively, we might think our model looks like the world, but when applied to less accessible parts of the world (like applying classical orbiting electrons to  the concept of atoms) things happen that don't make sense. Consequently, additional logic needs to be introduced to resolve such issues (like Bohr's model of the atom).

There is no reason why the universe should be logical, but we can use logic to understand the universe. It is important to distinguish between these two things. `Understanding' is always subjective: we come to understand the universe (or part of) in the context of a model. Any given model will have a specific structure and in the context of this structure (this is the subjectivity) we can ask questions and get answers. However, we are always restricted to only those questions that are relevant to the model, so we can only make deductions that are \emph{compatible} with our observations. The universe is not a space-time, but if we think about it in this way then we can obtain some understanding. Similarly, the universe is not a Hilbert space either. One of the reasons why the line between models and reality has become so blurred is due to our ability to engineer as a result of our acquisition of understanding: if a model can make accurate predictions then we conclude that we know how something works. If we know how something works then we can control it - build machines. If we can control the world we believe that our understanding is true. Moreover, we can use these machines to acquire more understanding of the world. In modern physics it is often the machines themselves that we are trying to understand. For example, we use our knowledge of quantum mechanics to build a sensitive apparatus to test the fundamental principles of quantum mechanics.

The basic idea in quantum mechanics is as follows. We consider an approximately isolated part of the universe, the `object', that is meaningful in the context of some apparatus that has been designed in accordance with prior knowledge of how the universe works based on a specific model through which we have some kind of context. This system may be a bit of matter or a diffraction grating with `known' properties and additional properties that are yet to be determined. As far as an experiment is concerned the system can interact with the apparatus in some known way. Again, this is entirely based on prior knowledge about what the apparatus is and does, and how we can control it.

It should now be coming to light that there are two main components in a complete theory of physics. The first is a mathematical structure that encodes, and attempts to capture, our observations and re-imagine them. The second is a mathematical structure that encodes, and attempts to capture, the means by which our observations arise. The usual idea is that there \emph{is} an objective reality, of which we are part, and life is a collection of experiences (observations) within this. A more factual approach would say that we have experiences and we are trying to give some kind of structure to whatever part of the universe our experiences are arising from. That is: we are trying to establish \emph{the source} of our experience of life. It tends to be easier to believe that this source has an objective reality, but even it doesn't we can simply define objective reality to be the emergent structure that we infer in order for us to understand our observations. Notice that this latter statement does not require objective reality to pre-exist, but rather, it implies that we are generating it on our journey of understanding. This also begs the question: can we choose our objective reality?

\section{From Classical to Quantum Logic}\label{2}
Classical logic is usually described by Boolean algebra, a distributive lattice with logical compliment. It consists of a  set $C$ of elements (events, propositions, knowledge, \emph{et c.}) $a,b,\ldots$ equipped with a partial order $\leq$, where $a\leq b$ means that if $b$ is `true' then it follows that $a$ is `true'. There is also a
negation which reverses the order relation such that $a\leq b\Leftrightarrow b^|\leq a^|$, where $a^|$ is the logical compliment of $a$ and $a^{||}=a$.  Further, there are binary operations $\wedge$ and $\vee$. The first is a logical infimum of $a$ and $b$, such that $a\wedge b$ is the most knowledge common to both $a$ and $b$. That is $a\wedge b\leq a$ and $a\wedge b \leq b$,  and we can also interpret $a\wedge b$ as the truth common to both $a$ and $b$. The second operation is the logical  supremum of $a$ and $b$, such that $a\vee b=(a^|\wedge b^|)^|$ is the least knowledge from which both $a$ and $b$ follow, so $a\vee b \geq a$ and $a\vee b \geq b$.

The set $C$ contains a greatest element $1=a\vee a^|$ for any element $a\in C$, which is `everything' (maximal knowledge) such that $1\wedge a=a$. Similarly, there is a least element $0=a\wedge a^|=1^|$, `nothing', such that $0\vee a=a$.  These binary operations are themselves associative and commutative and classical logic also imposes \emph{distributivity}:
\begin{equation}\label{dist}
    (a\wedge b)\vee c=(a \vee c)\wedge (b\vee c)\quad\text{and}\quad a\wedge (b\vee c)=(a \wedge b)\vee (a\wedge c).
\end{equation}
Two elements $a,b$ are called
\begin{equation}
    \text{\emph{inconsistent} if }a\wedge b=0,\quad \text{\emph{disjoint} if }b\leq a^|, \quad\text{and}\quad \text{\emph{conjoint} if }b^|\leq a.
\end{equation}
Inconsistent means that there is no knowledge following from $a$ which also follows from $b$, disjoint means that $b$ will be true if $a$ is false and vice versa, and conjoint means that if $a$ is true then $b$ is false (and vice versa) which also means that $a\vee b=(b^|\vee a)\vee b=a\vee 1=1$.
Notice that if $a$ and $b$ are both disjoint and conjoint then $a=b^|$ without the requirement for distributivity. However, the distributivity has a  very important consequence:
\begin{equation}
    \text{\emph{disjoint} $\Leftrightarrow$ \emph{inconsistent}}
\end{equation} \emph{Disjoint} $\Rightarrow$ \emph{inconsistent} simply follows from the associativity of $\wedge$, since $b\leq a^|\Rightarrow b=b\wedge a^|$ we have $0=0\wedge b=(a\wedge a^|)\wedge b=a\wedge (a^|\wedge b)=a\wedge b$. However,  \emph{inconsistent} $\Rightarrow$ \emph{disjoint} requires
distributivity: $b=(a\vee a^|)\wedge b=(a \wedge b)\vee ( a^|\wedge b)$ so $a\wedge b=0\Rightarrow b=a^|\wedge b\Rightarrow b\leq a^|$.

Additional structure that follows from distributivity is \emph{logical relativity} \cite{MQT}. This corresponds to defining a negation $b^{a|c}$ of an element $b$ with respect to an interval $b\in[a,c]\subset C$ which ultimately corresponds to the conditioning of $b$ with respect to $[a,c]:=\{e\in C: a\leq e\leq c\}$. It is given by the associative operation
\begin{equation}
    b^{a|c}=(a\vee b^| )\wedge c=a\vee( b^| \wedge c)\equiv a\vee b^| \wedge c,
\end{equation}
and defines an involution as $(b^{a|c})^{a|c}=a\vee (b^{a|c})^| \wedge  c = a\vee (c^|\vee b \wedge  a^|)\wedge  c =b$.

The simplest example of a Boolean algebra is $C=\{0,1\}$ but this is not so useful as we only have `nothing' (no knowledge) or `everything' (maximal knowledge). What's more interesting is $C=\{0,e,e^|,1\}$, \emph{the classical bit}, where $e\wedge e^|=0$ and $e\vee e^|=1$. This is a very simple and intuitive way of handling observations: an observation is simply represented by an event $e$. But note that the compliment of this event, $e^|$, is itself another event.

This classical logic has a concrete realisation as an algebra of projectors in an abelian $W^*$-algebra $\mathcal{C}$ (with Hermitian involution $\dag$)  so that $a\in C$ is mapped to a projection $P_a=P_a^\dag P_a=P_a^\dag\in\mathcal{C}$. The binary operations are given by  $P_{a\wedge b}=P_a P_b$ and $P_{a\vee b}=P_a-P_a P_b+P_b$, with $P_0=0$ and $P_1=I$, so that $P_{a^|}=I-P_a$, and inconsistence (which is also disjunction) is written as $P_a P_b=0$, whilst conjunction is $(I-P_a)P_b=I-P_a$, or $P_a-P_a P_b+P_b=I$. Further, the partial ordering $P_a\leq P_b$ is determined by the non-negativity of $P_b-P_a$ (non-negative spectrum), and the relative compliment of $b\in[a,c]$ is given by
$P_{b^{a|c}}=(P_a-P_a(I-P_b)+(I-P_b))P_c=P_c -(P_b-P_a)$.

Our classical way of thinking about this mathematics often involves us picturing sets and thinking of $\wedge$ as set intersection and $\vee$ as set union, but this imagination is distributive in its nature. Physics experiments have revealed that this classical logic is too restrictive to correctly explain our observations, but more experiments can be explained if we imagine the universe in a non-distributive way.
Quantum logic is non-distributive and a very reasonable generalisation of classical logic. The idea is this: if two events are inconsistent why should they have to be disjoint? This would mean that we could have events $a$ and $b$ from which  no common knowledge follows, yet these events are not disjoint.  Such propositions are called \emph{incompatible}.
\begin{Remark}In fact,   for $a$ and $b$ to be incompatible we don't have to have $a\wedge b=0$ but rather: $a$ and $b$ are incompatible if $\exists$ non zero $p\leq a$ and $q\leq b$ such that $p\wedge q=0$ and $p\nleq q^|$.\end{Remark}

 Note that if  inconsistent events were disjoint then distributivity, and thus classical logic, would follow as a special case. The problem we now face is that if we try to imagine the events $a$ and $b$ `visually' as sets then we can't distinguish inconsistence from disjunction.  However, the algebra of projectors allows us to handle this conceptually reasonable quantum logic as follows.    If $A$ is a quantum logic then it can contain classical logics $C\subset A$. Any such $C$ has a representation as a set of commuting projectors in a subalgebra $\mathcal{C}$ of a $W^*$-algebra $\mathcal{A}$, as described above. Only this time we can say that two elements $a$ and $b$ are disjoint iff $P_a P_b=0$, and in order to understand inconsistence we must understand that $P_a\wedge P_b:=P_{a\wedge b}$ is the largest projector such that $P_a-P_{a\wedge b}$ and $P_b-P_{a\wedge b}$ are non-negative. Similarly, $P_a\vee P_b:=P_{a\vee b}$ is the smallest projector such that $P_{a\vee b}-P_a$ and $P_{a\vee b}-P_b$ are non-negative. Note that, if we denoted by $\mathcal{H}$ the representing Hilbert space of $\mathcal{A}$ then an operator $X$ is non-negative if $\langle\psi|X|\psi\rangle\geq 0$ for all $|\psi\rangle\in\mathcal{H}$.

 Now suppose that we have two incompatible elements  $a$ and $b$, such that $P_{a\wedge b}=0$ and $P_a P_b\neq 0$,
 then it follows that $P_a P_b\neq P_b P_a$. To see this, suppose these projectors commute, $[P_a,P_b]=0$, then $P_{a\wedge b}=P_a P_b$. Then inconsistence has the form $P_a P_b=0$, which also implies that $a$ and $b$ are disjoint. Therefore, if $a$ and $b$ are not disjoint then $[P_a,P_b]\neq 0$.
  So, if we let go of the distributivity of the binary operations of classical logic, then we end up with the quantum logic of events described by projectors in a non-commutative $W^*$-algebra $\mathcal{A}$. In order to generate the algebra $\mathcal{A}$ for a given model of some specific observations that we make we must first identify specific  observations $a$ and identify each with a projector $P_a$. Then we get $\mathcal{A}=\{P_a\}''$ by taking the bicommutant of the set of all such projectors.  So, the algebra $\mathcal{A}$ is the algebra generated by some generally incompatible observations.

  The representing Hilbert space $\mathcal{H}$ makes it easier to understand $P_{a\wedge b}$ and $P_{a\vee b}$. The former is the projection onto the Hilbert subspace that is the intersection of the Hilbert subspaces $P_a\mathcal{H}$ and $P_b\mathcal{H}$, such that $P_{a\wedge b}\mathcal{H}=P_{a}\mathcal{H}\cap P_{ b}\mathcal{H}=\{\psi\in\mathcal{H}:P_a\psi=\psi=P_b\psi\}$. The latter is the projection onto the closed linear span of these Hilbert subspaces such that $P_{a\vee b}\mathcal{H}=P_{a}\mathcal{H}\cup P_{ b}\mathcal{H}:=\overline{P_a\mathcal{H}+P_b\mathcal{H}_b}=\overline{\{\psi=\psi_a+\psi_b: \psi_k\in P_k\mathcal{H}; k=a,b\}}=(I-P_{a^|\wedge b^|})\mathcal{H}$.

 The main point about this quantum logic is that it consists of incompatible classical logics. Each classical logic corresponds to some kind of `conventional' way of describing the world, but the `quantumness'  allows different ways of describing the same world. Such different descriptions of the world are incompatible in the sense that there may be nothing common to either description of the world, yet these two seemingly different descriptions may \emph{completely} describe the same thing. To be clear about the latter: `different descriptions of the same thing' does not mean two disjoint propositions in a classical logic, it means there can be two non-disjoint propositions in two inconsistent classical logics. A well-known example in physics is the description of a particle. One could \emph{completely} describe this system in either the position representation or  the momentum representation; these are two incompatible descriptions of the same system. It was originally believed that position and momentum formed independent parts of a system,
 but on closer inspection this is seen to be false.

\section{The Completion of Quantum Mechanics}\label{3}

A standard assumption in orthodox quantum mechanics is that the evolution of a quantum system is determined by a unitary (strongly continuous one parameter)  group generated by the Hermitian  Hamiltonian operator. However, in this section it shall be shown that when measurement is included in the standard quantum theory, such Hamiltonian evolution is insufficient to describe the dynamics of a measurement apparatus coupled to a quantum system. Thus the Schr\"odinger equation is an incomplete description of quantum dynamics. Indeed, the complete description shall be derived and it shall be shown to correspond to Belavkin's Quantum Filtering Theory.   Intuitively, quantum filtering is a dynamical theory of information-extraction from a quantum system $\mathfrak{h}$. Moreover, it respects the unitary pure-state dynamics of wave-functions. Now the intention is to shown that this theory is not just {a} solution to the quantum measurement problem, but \emph{the} solution.

\subsection{von Neumann's Measurement Theory}
The first attempt to complete quantum mechanics was done in 1932 by von Neumann \cite{vN} and at first glance this appears to be quite reasonable. Following ideas that he also attributes to Szilard and Heisenberg, he divides the world into three parts, I, II and III, where I is the quantum system under observation, II is the measurement instrument and III is the actual observer. Different interpretations of II and III were considered, with the extreme case being that II is everything that connects the system I with the observer's abstract `ego' III. However, for practical purposes many of the concepts contained within II were cast into III, leaving only some basic type of apparatus to constitute II. For example, in the case of Schr\"odinger's cat that would mean taking I to be a two-level atom, II to be the cat (serving as an apparatus) and III to be the observer's ego and physical means by which they observe the state of the cat. All that said, von Neumann then excludes III from the calculations, with the justification that it makes no difference to predictions about I, and proceeds as follows.

The initial state of the composite system I + II may be given by a wave-function $\psi=|\psi\rangle\otimes |\xi_0\rangle$ in a separable Hilbert space $\mathcal{H}=\mathfrak{h}\otimes\mathfrak{k}$, where $\mathfrak{k}=L^2(\Omega)$ is the instrument Hilbert space constituting an apparatus and $\Omega$ is the set of measurement  outcomes $k$. Meanwhile,  $\mathfrak{h}$ is the Hilbert space of the quantum system under observation. A measurement first involves an interaction between the system and instrument given by a unitary operator $U(t)=e^{-\mathrm{i}Ht}$, with Hamiltonian $\hbar H$, coupling  (entangling) the instrument with the system over an interval of time $[0,t]$ to give $\psi(t)=U(t)\psi$. Following this interaction a measurement may be performed by the instrument resulting in an observation of type-$k$, with respect to some chosen orthonormal basis $\{|\xi_k\rangle\}$ for $\mathfrak{k}$. The entanglement of the system, which we'll call `object', with the apparatus may be written in the form
\begin{equation}\label{int}
    \psi(t)=U(t)(|\psi\rangle\otimes |\xi_0\rangle)
    =\sum V_{k}(t)|\psi\rangle\otimes|\xi_k\rangle,
\end{equation}
where $V_k(t)=(1_\mathfrak{h}\otimes\langle\xi_k|)e^{-\mathrm{i}Ht}(1_\mathfrak{h}\otimes|\xi_0\rangle)$ are contractions (norm reducing operators) in $\mathcal{B}(\mathfrak{h})$ with $V_k(t)|\psi\rangle:=c_k(t)|\psi_k(t)\rangle$,
where $\sum|c_k|^2=1$ and the $|\psi_k\rangle$ are normalised but not necessarily orthogonal.
The correct form of the Born rule would then state that the probability of making an observation of type-$k$ is given by $|c_k|^2=\psi(t)^\dag(1_\mathfrak{h}\otimes P_k)\psi(t)$, where $P_k=|\xi_k\rangle\langle\xi_k|$. This is the `correct' expression because the probabilities associated with an observation should be determined from the observation itself (which is given by the apparatus projector $P_k$) which in turn leads to inference about the quantum system.

In fact, von Neumann did not give the interaction in the form \eqref{int}, but instead in the form
\begin{equation}
    U(t)(\sum c_k|\phi_k\rangle\otimes|\xi_0\rangle)=\sum c_k|\phi_k\rangle\otimes|\xi_k\rangle
\end{equation}
corresponding to the realisation of the system in a state given by the vector $|\phi_k\rangle\in\mathfrak{h}$
 supposed to be a normalised eigenvector of whatever observable in $\mathcal{B}(\mathfrak{h})$ is being measured (i.e. the case where $V_k(t)$ is a projector $E_k$), but this did not respect the assumption that $U(t)=e^{-\mathrm{i}Ht}$.
So, in order to reconcile this interaction with the Schr\"odinger equation
von Neumann considered, instead, the case where $\mathfrak{h}$ and $\mathfrak{k}$ are both copies of $L^2(\mathbb{R})$ and then considers an  interaction of the form $U(t)=e^{\mathrm{i} (X\otimes D) t}$, neglecting kinetic energy for simplicity, where $X$ is the position operator of the system under observation and $\hbar D$ is the instrument momentum operator.
Now \eqref{int} assumes the form
\begin{equation}\label{int2}
    U(t)(|\psi\rangle\otimes |\xi_0\rangle)=\iint \psi(x)\xi_0(y+ t x)|x\rangle\otimes |y\rangle \mathrm{d}x\mathrm{d}y
\end{equation}
with respect to Dirac's generalised eigenfucntions $|x\rangle\notin\mathfrak{h}$, where $X|x\rangle =x|x\rangle$ and $\langle x|x'\rangle=\delta(x-x')$.

The interaction \eqref{int} satisfies the Schr\"odinger equation but there are two problems hidden in this theory. The first is that the dynamics of the measurement itself, regarded as the spontaneous action of a projector $P$, is not treated formally, and the second is that when we try to condition future expectations of an observable on prior observations the Law of Total Probability is not considered.  Now these issues shall be addressed in detail, and the consequences thereof.

\subsection{Non-commutative Conditional Expectations and Observation}

 The quantum filtering theory, amongst other things, is also a theory of quantum causality and we shall see how its structure may be derived from some logical principles underlying observation. Ultimately we must understand how to condition future observations on past observations, and then understand how to construct a time-continuous process of this kind. So quantum conditional expectations and stochastic calculus will play a fundamental role here. A very vague history of these things is as follows, which will doubtless do a lot of injustice to the many minds that have been involved.

 Stochastic calculus did not emerge until It\^o's inventions in the 1940's and, according to Umegaki \cite{UM}, there was no proper theory of quantum probability, motivating him to construct a non-commutative conditional expectation in the 1950's. But since these inventions both quantum probability and stochastic calculus have been developed extensively; see citations in \cite{CHAO}. Further, a solution to the classical filtering problem emerged in the 1960's from Stratonovivch \cite{ST}, going into the 1970's quantum stochastic calculus developed from the study of quantum dynamical semigroups,  see \cite{D} for example, resulting in the famous Lindblad equation,
 and in the 1980's this developed into Hudson's and Parthasarathy's rigorous Fock space formalism of quantum stochastic differential equations \cite{HP}. Meanwhile, Stratonovich's student Belavkin, who was also a pioneer of quantum information and quantum stochastic calculus,  was developing quantum filtering. However, others were also working on the theory of time-continuous quantum measurement, see citations in \cite{CTI}.

\begin{Definition} \cite{KAD}\label{def1}
    A non-commutative conditional expectation is a positive, unital, linear map between von Neumann algebras $\epsilon:\mathcal{A}\rightarrow\mathcal{C}$, where $\mathcal{C}\subset\mathcal{A}$ and $\epsilon(C_1 A C_2)=C_1\epsilon(A)C_2$ for all $C_1,C_2\in\mathcal{C}$ and $A\in\mathcal{A}$. 
\end{Definition}
We shall begin with some basic examples of conditional expectations on the von Neumann algebra $\mathcal{A}=\mathcal{B}(\mathcal{H})$, for a separable Hilbert space $\mathcal{H}$, and consider a projection $P\in\mathcal{A}$.
Now define $\mathcal{C}\subset \mathcal{A}$ as the von Neumann algebra generated by $P$ and its orthogonal compliment $Q=I-P$ such that $\mathcal{C}=\{c_1 P+c_2 Q:c_1,c_2\in\mathbb{C}\}$, and note that it is abelian (commutative). First, we shall define the conditional expectation $\omega:\mathcal{A}\rightarrow \mathcal{C}'$, where $\mathcal{C}'=\{A\in\mathcal{A}:[A,C]=0\;\forall\;C\in\mathcal{C}\}=\{PAP+QAQ:A\in\mathcal{A}\}$ is the commutant of $\mathcal{C}$ in $\mathcal{B}(\mathcal{H})$ and generally not abelian. This conditional expectation is defined as
\begin{equation}\label{CEC}
    \omega(A)=PAP+QAQ
\end{equation}
and needn't be referred in any way  to a state $\mathbb{E}$. It is well-defined as a non-commutative conditional expectation according to Definition \ref{def1}. However, if we now define a conditional expectation $\epsilon:\mathcal{A}\rightarrow\mathcal{C}$ then we must introduce a state $\mathbb{E}$ on $\mathcal{A}$ as a means to reduce operators to scalars. Then we can define
\begin{equation}\label{CE}
    \epsilon(A)=\frac{\mathbb{E}[PAP]}{\mathbb{E}[P]}P+\frac{\mathbb{E}[QAQ]}{\mathbb{E}[Q]}Q.
\end{equation}
We shall soon see that if $\mathbb{E}$ is a normal state, given by a trace-class density operator so that $\mathbb{E}[A]=\texttt{Tr}[\varrho A]$, then $\varrho$ must generally assume the form $P\psi\psi^\dag P+Q\psi\psi^\dag Q\in\mathcal{C}'$.

Let's now apply this to von Neumann's measurement theory,  recalling the object-apparatus interaction given by \eqref{int2}, where $\mathcal{H}=\mathfrak{h}\otimes\mathfrak{k}$ with $\mathfrak{h}=\mathfrak{k}=L^2(\mathbb{R})$, and see what happens when we take a measurement at a time $s$, say. The measurement is represented by a projector $P$ and  the von Neumann-L\"uders projection postulate, written in its proper form, states that
\begin{equation}\label{pp}
   \psi(s)\mapsto \frac{(1_\mathfrak{h}\otimes P)\psi(s)}{\|(1_\mathfrak{h}\otimes P)\psi(s)\|}
\end{equation}
and consequently this new wave-function continues to evolve according to the interaction Hamiltonian from $s$ up to some later time $t>s$ if the apparatus and object are assumed to remain coupled, or re-couple, after the measurement. As a simple example we could suppose the projector $P$ in \eqref{CE} represents an observation that a quantum particle is in a region $\Delta\subset \mathbb{R}$, so that $Q$ represents the observation that the particle is not in $\Delta$. Then, with $\mathbb{E}_s[A]:=\psi(s)^\dag A\psi(s)$ and $A=X\otimes 1_\mathfrak{k}$, the conditional  expectation \eqref{CE} would assume the form
\begin{equation}\label{ob11}
\epsilon_s(X)=\frac{\mathbb{E}_s[X\otimes P]}{\mathbb{E}_s[1_\mathfrak{h}\otimes P]}P+\frac{\mathbb{E}_s[X\otimes Q]}{\mathbb{E}_s[1_\mathfrak{h}\otimes Q]}Q.
\end{equation}
It is important to understand that we are imposing the condition that the particle will be in \emph{either} $\Delta$ \emph{or} $\mathbb{R}\setminus\Delta$ when observed at time $s$, and not in a superposition of these positions. Such is a consequence of the experimental setup that is designed to do just this. In particular, we are aiming to make predictions about some `quantum property' $X$ at time $t$ based on the outcome of an observation at an earlier time $s$, but we must first understand these conditional expectations in more detail.
\begin{Remark}
    In the study of von Neumann algebras Hermitian operators are called \emph{symmetric} and such an $X$, defined on a dense domain $\mathfrak{d}_X\subset\mathfrak{h}$, has $\mathfrak{d}_X\subseteq\mathfrak{d}_{X^\dag}$ and satisfies $X^\dag|\psi\rangle=X|\psi\rangle$ for all $|\psi\rangle\in\mathfrak{d}_X$. If a symmetric operator $X$ is unbounded but $(X+\mathrm{i}I)^{-1}\in\mathcal{B}(\mathfrak{h})$ then although $X\notin\mathcal{B}(\mathfrak{h})$ it may still be said that $X$ is \emph{affiliated} to $\mathcal{B}(\mathfrak{h})$ and for most practical purposes $X$ may be handled in the same way as a bounded operator \cite{LUC}. Here we shall denote this affiliation by $X\neg\mathcal{B}(\mathfrak{h})$ and we shall also include within such terminology any self-adjoint $X\in\mathcal{B}(\mathfrak{h})$.

\end{Remark}
\begin{Definition}
    A quantum property is an essentially  self-adjoint operator $X$ affiliated to a von Neumann algebra. That means $X$ is a symmetric operator where the ranges of the operators $X\pm\mathrm{i}I$ are dense in $\mathfrak{h}$ \cite{RS}.
\end{Definition}

We can see that measurements and their outcomes are intimately related to  conditional expectations of quantum properties. However, the concept of such a conditional expectation is not yet complete, for we must insist that quantum properties be \emph{predictable}. This means that  the \emph{Law of Total Probability}   \eqref{LTP} must be satisfied, and that is the insistence  that the expectation of a quantum property is not affected by the conditioning.
\begin{Definition}
A non-commutative conditional expectation $\epsilon:\mathcal{A}\rightarrow\mathcal{C}$, as given in Definition \ref{def1}, defines  the {conditional expectation with respect to a normal state $\mathbb{E}$},  of an operator $A\neg\mathcal{A}$ with $\mathbb{E}[A]<\infty$,  if
    \begin{equation}\label{LTP}
       \mathbb{E}[ \epsilon(A)]=\mathbb{E}[A].
    \end{equation}
In which case we can write $\epsilon(A)=\mathbb{E}[A|\mathcal{C}]$. The identity \eqref{LTP} is called the Law of Total Probability (LTP) and if satisfied by $A$ we say that $A$ is predictable with respect to $\mathbb{E}|\mathcal{C}$, or $\mathbb{E}|\mathcal{C}$-predictable. If \eqref{LTP} holds for all $A\in\mathcal{A}$ we say that $\mathcal{A}$ is $\mathbb{E}|\mathcal{C}$-predictable.
\end{Definition}
Notice that if LTP is to be satisfied by \eqref{CE} then we must have
\begin{equation}\label{ltpeg}
    \mathbb{E}[PAP+QAQ]=\mathbb{E}[A],
\end{equation}
and the same can be said for \eqref{CEC} if we were to introduce a state.
Moreover, \eqref{ltpeg} leads us to one of two conditions. The first is that $[P,A]=0$ and the second is that $[P,\varrho]=0$. We shall soon see that these amount to the same thing.

\begin{Definition}\label{obom}
Let $\mathcal{A}$ be a von Neumann algebra and consider a quantum property $X\neg \mathcal{B}(\mathfrak{h})$, then we shall say that $X$ is $\Omega$-observable, observable with respect to a set of measurement outcomes $\Omega$, if $X$ is predictable with respect to $\mathbb{E}|\mathcal{C}$ where $\mathcal{C}=\{P_k:k\in\Omega\}''\subset\mathcal{B}(\mathfrak{h})'$ is the abelian algebra generated by a complete set of orthogonal projections $P_k\in\mathcal{A}$, $\sum P_k=I$, which represent observable events (measurement outcomes) $k\in\Omega$, such that $\mathcal{C}$ is the representation of a classical logic of measurement outcomes.
\end{Definition}
The important point here is that  $\mathcal{C}\subset\mathcal{B}(\mathfrak{h})'$, then since $\mathcal{B}(\mathfrak{h})$ is a type-$I$ factor it follows that $\mathcal{A}=\mathcal{B}(\mathfrak{h})\otimes \mathcal{C}$. Moreover, the conditional expectation $\epsilon(X)=\mathbb{E}[X|\mathcal{C}]$ is given by
\begin{equation}
    \epsilon(X\otimes1_\mathfrak{k})=\sum \frac{ \mathbb{E}[X\otimes P_k]}{\mathbb{E}[1_\mathfrak{h}\otimes P_k]}P_k,
\end{equation}
where $\mathcal{C}\subset\mathcal{B}(\mathfrak{k})$ and the sum is taken over all $k\in\Omega$ for which $\mathbb{E}[1_\mathfrak{h}\otimes P_k]\neq0$, and it should be clear that LTP is satisfied. We shall often denote $\epsilon(X\otimes1_\mathfrak{k})$ simply as $\epsilon(X)$ for brevity.
This may seem like a big assumption to make about the mathematical structure of a quantum system under observation, but we'll now see this follows from the LTP for any $X\neg\mathcal{B}(\mathfrak{h})$ conditioned on projectors $E_k\in\mathcal{B}(\mathfrak{h})$ which are complete in the sense that $\sum E_k=1_\mathfrak{h}$.
\begin{Lemma}\label{lem1}
    Suppose that $X\neg\mathcal{B}(\mathfrak{h})$ is predictable with respect to a normal state $\mathbb{G}$ and a conditional expectation $\omega:\mathcal{B}(\mathfrak{h})\rightarrow \{E_k:k\in\Omega\}''\subset\mathcal{B}(\mathfrak{h})$, where $\{E_k\}$ is a complete set of orthoprojectors in $\mathcal{B}(\mathfrak{h})$. Then there is an abelian algebra $\mathcal{C}=\{P_k:k\in\Omega\}''$ represented on a Hilbert space $\mathfrak{k}$, where $\{P_k\}$ is a complete set of orthoprojectors in $\mathcal{B}(\mathfrak{k})$, such that $\mathbb{G}$ admits a purification $\mathbb{E}$ on the von Neumann algebra $\mathcal{C}'=\mathcal{B}(\mathfrak{h})\otimes\mathcal{C}:=\mathcal{A}$ and there is a conditional expectation $\epsilon:\mathcal{A}\rightarrow\mathcal{C}$ called the canonical representation of $\omega$. That is, if $X$ is $\mathbb{G}|\{E_k\}''$-predictable in $\mathcal{B}(\mathfrak{h})$ then it is $\Omega$-observable (i.e. $\mathbb{E}|\{P_k\}''$-predictable in $\mathcal{B}(\mathfrak{h})\otimes\mathcal{C}$).
\end{Lemma}

\begin{proof}
If $X\neg\mathcal{B}(\mathfrak{h})$ is $\mathbb{G}|\{E_k:k\in\Omega\}''$-predictable, with projectors $E_k\in\mathcal{B}(\mathfrak{h})$, $\sum E_k=1_\mathfrak{h}$, then there is a conditional expectation
\begin{equation}\label{omce}
    \omega(X)=\sum \frac{ \mathbb{G}[E_kXE_k]}{\mathbb{G}[E_k]}E_k,
\end{equation}
where sum is taken over non-zero $\mathbb{G}[E_k]$. Since $\mathbb{G}$ is assumed to be a normal state it is given by a density operator $\varrho$ and LTP \eqref{LTP} requires that $\texttt{Tr}[\varrho X]=\texttt{Tr}[\sum\varrho E_kXE_k]$. If $[X,E_k]\neq0$ then we must have $\sum E_k\varrho E_k=\varrho$ and this means that there is a Hilbert space $\mathfrak{k}$ and projectors $P_k\in\mathcal{B}(\mathfrak{k})$ such that the density may be decomposed as $\varrho=\widetilde{\psi}\widetilde{\psi}^\dag$ where $\widetilde{\psi}:\mathfrak{k}\rightarrow\mathfrak{h}$ and $E_k\widetilde{\psi}=\widetilde{\psi}P_k^\intercal$, where $P^\intercal$ is operator transpose which may be induced by the anti-linear complex conjugation operator $C$ as $P^\intercal=C P^\dag C$.
To see this, first write $\varrho=\sum_k \sum_{i_k}|c_{i_k}|^2|\phi_{i_k}\rangle\langle\phi_{i_k}|$, where $k\in\Omega$, $i_k\in \texttt{I}_k$, $|\texttt{I}_k|\geq1$
if $\mathbb{G}[E_k]\neq 0$,  and $E_l|\phi_{i_k}\rangle=\delta_{lk} |\phi_{i_k}\rangle$, such that $|\phi_{i_k}\rangle$ are the (normalised) degenerate eigenvectors of the projectors $E_k$. Then define $\mathfrak{k}=L^2(\texttt{I})$ with orthonormal basis $\{|\xi_{i}\rangle\}$, where $\texttt{I}=\sqcup\texttt{I}_k$ and $\sqcup$ denotes disjoint union, and define $\widetilde{\psi}=\sum _\texttt{I}c_{i}|\phi_i\rangle\langle \bar\xi_i|$, where $\langle \bar\xi_i|=\langle\xi_i|C$. Then $P_k=\sum_{\texttt{I}_k}|\xi_i\rangle\langle\xi_i|$ satisfies $E_k\widetilde\psi=\widetilde \psi P_k^\intercal$ for all $k\in\Omega$. If $\mathbb{G}[E_k]=0$ for a $k$ then $E_k|\phi_i\rangle=0$ for any $i\in\texttt{I}$. Then one may simply add an extra dimension to $\mathfrak{k}$ and define a corresponding projector $P_k=0\oplus1\in\mathfrak{k}\oplus \mathbb{C}$.

The operator $\widetilde\psi$ is the partial transpose of an entangled vector $\psi\in\mathcal{H}=\mathfrak{h}\otimes\mathfrak{k}$, and this partial transposition
is defined on product vectors as
\begin{equation}
\widetilde{|\phi\rangle\otimes|\xi\rangle}=|\phi\rangle\otimes |\xi\rangle^\intercal=|\phi\rangle\langle\bar\xi|,
\end{equation}
where $\langle\bar\xi|=\langle\xi|C$. Indeed, $\psi$ simply corresponds to the purification of $\varrho$ and defines the pure-state $\mathbb{E}$. Now on $\mathcal{H}$ we see that
\begin{equation}
(E_k\otimes 1_\mathfrak{k})\psi=\widetilde{E_k\widetilde\psi}= \widetilde{\widetilde\psi P_k^\intercal}=(1_\mathfrak{h}\otimes P_k)\psi,
\end{equation}
so we see that $\mathbb{G}[E_kXE_k]=\mathbb{E}[X\otimes P_k]$ so \eqref{omce} can be represented in $\mathcal{C}=\{P_k:k\in\Omega\}''$ in the form
\begin{equation}\label{epce}
    \epsilon(X\otimes 1_\mathfrak{k})=\sum \frac{ \mathbb{E}[X\otimes P_k]}{\mathbb{E}[1_\mathfrak{h}\otimes P_k]}P_k,
\end{equation}
where $\mathcal{C}$ is isomorphic to $\{E_k:k\in\Omega\}''$ due to the one-to-one correspondence $E_k\leftrightarrow P_k$ between the projectors $E_k\in\mathcal{B}(\mathfrak{h})$ and $P_k\in\mathcal{B}(\mathfrak{k})$. Note that $\mathfrak{k}$ and $\mathfrak{h}$ need not have the same dimension. 

In the special case where $[E_k,X]=0$ we can take $\varrho$ to be pure, given by $|\psi\rangle$ but we can trivially obtain the same dilated structure nonetheless. This is exactly the same procedure as the $[E_k,X]\neq0$ case, but this time we must transform the density $|\psi\rangle\langle\psi|$ into the mixture $\sum E_k|\psi\rangle\langle\psi|E_k$ first. This is then purified to $\psi=\sum c_k|\phi_k\rangle\otimes|\xi_k\rangle$ in the Hilbert space $\mathfrak{h}\otimes L^2(\Omega)$, where $E_k|\psi\rangle=c_k|\phi_k\rangle$.
\end{proof}

Now we see that if a quantum property $X$ is observable with respect to some measurement outcomes $\Omega$ then it is actually the measurement outcomes, given by projectors $P_k$, which are the `true' observables. That is, the projectors $P_k$ are the mathematical representation of the actual observations that can be made, and the quantum property is being observed \emph{indirectly} by virtue of its coupling to the measuring instrument. This is similar to von Neumann's measurement theory, only here the tensor structure has been derived from a requirement to satisfy LTP.

Notice that in Lemma \ref{lem1} we have derived the canonical structure of observation of a quantum system. However, this was derived from an attempt to take projections in $\mathcal{B}(\mathfrak{h})$. If, instead, we start with the tensor product structure $\mathfrak{h}\otimes\mathfrak{k}$ then we still have conditional expectations of the form \eqref{epce}, but now $(1_\mathfrak{h}\otimes P_k)\psi=V_k|\psi\rangle\otimes |\xi_k\rangle$, as mentioned below \eqref{int}, where the operators $V_k$ needn't be projectors $E_k$; the only requirement is that $\sum V_k^\dag V_k^{}=1_\mathfrak{h}$. This means that $\mathbb{E}[X\otimes P_k]=\mathbb{G}[V_k^\dag X V_k^{}]$ where $\mathbb{E}$ is the pure-state given by the vector $\psi=U(|\psi\rangle\otimes|\xi\rangle)\in\mathfrak{h}\otimes\mathfrak{k}$ and $\mathbb{G}$ is a pure-state given by $|\psi\rangle$; note that $V_k=(1_\mathfrak{h}\otimes \langle \xi_k|)U(1_\mathfrak{h}\otimes |\xi\rangle)$. Since the vectors $|\psi_k\rangle=V_k|\psi\rangle$ are not generally orthogonal we cannot identify the projectors $P_k$ with any corresponding projectors $E_k\in\mathcal{B}(\mathfrak{h})$.
This is why the projection postulate must be written in the form \eqref{pp}. Further, if we start with only $\mathcal{B}(\mathfrak{h})$, then we cannot consider such cases when $|\Omega|>\dim\mathfrak{h}$. So now we'll try to derive the object-apparatus $\otimes$-structure for an $\mathbb{E}|\mathcal{C}$-predictable quantum system $\mathcal{B}(\mathfrak{h})$ without first attempting to condition with projectors $E_k\in\mathcal{B}(\mathfrak{h})$.

\begin{Theorem}\label{th1}
    Suppose that a quantum system is represented by $\mathcal{B}(\mathfrak{h})$ and is predictable with respect to an abelian von Neumann algebra $\mathcal{C}$ and a state $\mathbb{E}$ on $\mathcal{A}:=\mathcal{B}(\mathfrak{h})\vee\mathcal{C}$, which is the smallest von Neumann algebra containing both $\mathcal{B}(\mathfrak{h})$ and $\mathcal{C}$. Then $\mathcal{A}=\mathcal{B}(\mathfrak{h})\otimes\mathcal{C}$.
\end{Theorem}
\begin{proof}
   Let $\mathcal{B}(\mathfrak{h})$ be $\mathbb{E}|\mathcal{C}$-predictable, where $\mathbb{E}$ is a normal state given by a density $\varrho$ and $\mathcal{C}$ is the abelian algebra generated by a complete set of projectors $\{P_k:k\in\Omega\}$, then the projectors $P_k$ are either in $\mathcal{B}(\mathfrak{h})$ or in $\mathcal{B}(\mathfrak{h})'$  this means we can  decompose $\{P_k:k\in\Omega\}=\{P_i:i\in\Omega_I\}\cup\{P_j:j\in\Omega_J\}$ where $P_i\in\mathcal{B}(\mathfrak{h})$ and $P_j\in\mathcal{B}(\mathfrak{h})'$, noting that LTP requires $[P_i,\varrho]=0$.
   Now let $\sum_iP_i=P_I$ and $\sum_jP_j=P_J$, then by  completeness $P_I+P_J=1_\mathcal{A}$ and we can write $\mathcal{A}=P_I\mathcal{A}+P_J\mathcal{A}$. Since $P_J\in\mathcal{B}(\mathfrak{h})'$ it follows that $P_jX=X\otimes P_j$ for any $X\in\mathcal{B}(\mathfrak{h})$ and $j\in\Omega_J$, so we can write $P_J\mathcal{A}=\mathcal{B}(\mathfrak{h})\otimes\mathcal{C}_J$, where $\mathcal{C}_J=\{P_j:j\in\Omega_J\}''$, noting that $P_j\in\mathcal{A}'$ too. Since  $P_I\mathcal{A}+P_J\mathcal{A}=\mathcal{A}P_I+\mathcal{A}P_J$, and $P_J\mathcal{A}=\mathcal{A}P_J$, it follows that $P_I\in\mathcal{A}'$, even though $P_i\notin\mathcal{A}'$. But $P_I\in\mathcal{B}(\mathfrak{h})$ so it must either be zero or the identity. If it is zero then we are already done. If it is the identity then Lemma \ref{lem1} concludes the proof.
\end{proof}
The main point about constructing Theorem \ref{th1} is simply to say that the observation of a quantum system \emph{always} has the $\Omega$-observable structure given in Definition \ref{obom}.
Let's return to the von Neumann measurement scheme as described in \eqref{ob11} and suppose we wished to estimate a quantum property $X$ at some time $t$ given an observation at some earlier time $s<t$, then we'd condition $X$ on the two possible wave-functions: $U(r)P\psi(s)$ and $U(r)Q\psi(s)$, where $\psi(s)=U(s)\psi$ and $U(t)=e^{-\mathrm{i}Ht}$, so that
\begin{equation}\label{CE2}
\epsilon_t(X,s)=\frac{\mathbb{E}_s[PU(r)^\dag X U(r)P]}{\mathbb{E}_s[P]}P+\frac{\mathbb{E}_s[QU(r)^\dag X U(r)Q]}{\mathbb{E}_s[Q]}Q
\end{equation}
is the conditional expectation of $X$ at a time $t=r+s>s$, given a complete set of observations, $P+Q=1_\mathfrak{k}$, at time $s$. Here we drop some of the tensor product notation for convenience, but note that $PU(r)^\dag X U(r)P$ should be written more precisely, but rather tediously, as $(1_\mathfrak{h}\otimes P)U(r)^\dag (X\otimes 1_\mathfrak{k}) U(r)(1_\mathfrak{h}\otimes P)$.  However, this conditional expectation will not satisfy LTP with respect to $\mathbb{E}_s$ if $\mathbb{E}_s$ is a pure-state. This is because  $P\in\mathcal{C}$ but $U(r)^\dag (X\otimes 1_\mathfrak{k})U(r)\notin\mathcal{C}'$, so $[1_\mathfrak{h}\otimes P,U(r)^\dag (X\otimes 1_\mathfrak{k})U(r)]\neq 0$. So, in order to have predictability Theorem \ref{th1} tells us that the system must be dilated again, as we shall see next.

\subsection{Trajectories, Causality and Quantum Filtering}

If we are to study quantum systems then we'd like to be able to predict their behaviour. That is, given an observation of the system at some time $s$ we'd like to predict the values of its properties at some later time $t\geq s$.
This leads us to the so called \emph{Non-demolition Principle of Quantum Causality} \cite{ND,NDP}, which, in a simplified  form based on our discussion so far, can be written as follows.
\begin{Corollary}\label{NDP}
A quantum property $X\neg\mathcal{B}(\mathfrak{h})$ is $\mathbb{E}_s|\mathcal{C}_s$-predictable at a time $t$, with respect to an observation at a time $s\leq t$, if and only if $[X_s^t,\mathcal{C}_s]=0$, i.e. $X_s^t\in\mathcal{C}_s'$ where $\mathcal{C}_s$ is the algebra of observations at  time $s$ and $X_s^t={U_s^t}^\dag X U_s^t$, where $U_s^t$ is a unitary evolution operator from $s$ up to $t$.
\end{Corollary}
\begin{proof}
    Follows immediately from Theorem \ref{th1}.
\end{proof}
Here the notation $U_s^t$ has been introduced to allow for more general evolution operators in what follows. They are required to satisfy the cocycle property $U_r^t=U_s^t U_r^s$, $\forall\; t>s>r$, and the evolution $U_s^t=U(t-s)$ where $U(r)=e^{-\mathrm{i}Hr}$ is a special case.
\begin{Corollary}\label{cor2}
      Let  $X\neg\mathcal{B}(\mathfrak{h})$ be $\mathbb{E}_s|\mathcal{C}_s$-predictable at a time $t>s$, where $\mathcal{C}_s\subset\mathcal{B}(\mathfrak{k})$ is the algebra of observation at a time $s$, then $X$ is $\mathbb{E}_t|\mathcal{C}_s$-predictable.
\end{Corollary}
\begin{proof}
    By Corollary \ref{NDP} the commutativity $[X_s^t,P_k]=0$, for $P_k\in\mathcal{C}_s$, means that $X_s^t$ and $P_k$ admit the factorisations $X_s^t\otimes 1_{\mathfrak{k}}$ and $1_\mathfrak{g}\otimes P_k$ such that the representing Hilbert space $\mathcal{H}$ admits the factorisation $\mathfrak{g}\otimes\mathfrak{k}$, where $X_s^t\neg\mathcal{B}(\mathfrak{g})$ with $\mathfrak{g}\supseteq\mathfrak{h}$. Moreover, $U_s^t\in\mathcal{B}(\mathfrak{g})$.
    Now, if we denote the unitary evolution of the total system by $U_0^t\psi$, then we see that an observation at time $s<t$ now implies the factorisation $U_0^t=(U_s^t\otimes 1_{\mathfrak{k}})U_0^s$.
    Thus $\mathbb{E}_s[X_s^t \otimes P_k]=\mathbb{E}_t[X\otimes P_k]$.
\end{proof}

\begin{Lemma}\label{lem2}
    Let $X\neg\mathcal{B}(\mathfrak{h})$ be $\mathbb{E}_t|\mathcal{C}$-predictable where $\mathcal{C}$ is the abelian von Neumann algebra representing the observations at times $r$ and $s$ with $t>s>r>0$. Then the representing Hilbert space admits the factorisation $\mathcal{H}=\mathfrak{g}\otimes\mathfrak{k}_s\otimes\mathfrak{k}_r$, where $\mathfrak{h}\subseteq\mathfrak{g}$.
\end{Lemma}
\begin{proof}
    Consider wave-function measurement dynamics of the form $U_s^t P_{k_2}U_r^s P_{k_1}U_0^r\psi$, with respect to an observation $k_1\in\Omega_1$ followed by $k_2\in\Omega_2$, where $\psi=|\psi\rangle\otimes\Phi$. Although the evolution operator $U_r^t=U_s^tU_r^s$ must commute with $P_{k_1}$, in particular $[U_r^s,P_{k_1}]=0$, but the same is not true for $U_0^r$. In fact, we must have $[P_{k_1},U_0^r]\neq 0$ because $U_0^r$ describes the coupling between the quantum object and the apparatus, up to time $r$. Similarly, $[U_s^t,P_{k_2}]=0$ and $[P_{k_2},U_0^s]\neq 0$, where $U_0^s=U_r^sU_0^r$. In particular, $[P_{k_2},U_r^s]\neq 0$. Now it follows that $U_s^t P_{k_2}U_r^s P_{k_1}=P_{k_2}U_s^tU_r^s P_{k_1}=P_{k_2} P_{k_1}U_s^tU_r^s$, but also $U_s^t P_{k_2}U_r^s P_{k_1}=P_{k_1}U_s^t P_{k_2}U_r^s =P_{k_1} P_{k_2}U_s^tU_r^s$. This must hold in general, so we see that $[P_{k_1},P_{k_2}]=0$. That is, if $\mathcal{C}=\mathcal{C}_r\vee\mathcal{C}_s$ then $\mathcal{C}_s\subset\mathcal{C}_r'$, and due to the arbitrariness of the projectors' bases it follows that $\mathcal{C}=\mathcal{C}_s\otimes\mathcal{C}_r$. Alternatively, notice that we can write  $U_r^s\in\mathcal{B}\subset\mathcal{C}_r'$, so $\mathcal{C}_r\subset\mathcal{B}'$, and since $[P_{k_2},U_r^s]\neq 0$ we see $P_{k_2}\notin \mathcal{B}'$ and thus $P_{k_2}\notin\mathcal{C}_r$.
 \end{proof}

The consequence of this is that each instance of an observation must correspond to a new copy of the measuring instrument, collectively forming an apparatus. So now let's consider a chain (ordered set) of $n$ measurements at times
\begin{equation}
    \vartheta=\{t_n>t_{n-1}>\cdots>t_1\},\quad|\vartheta|=n
\end{equation}
and so the dynamics of the object-apparatus wave-function in the Schr\"odinger picture has the {general} form
\begin{equation}\label{gm1}
    \psi_t(\vartheta)=(U_{t_n}^t\otimes 1_n\otimes1_{n-1}\otimes\cdots\otimes 1_1)(U^{t_n}_{t_{n-1}}\otimes 1_{n-1}\otimes\cdots \otimes1_1)\cdots U_0^{t_1}\psi
\end{equation}
in the Hilbert space $\mathcal{H}(\vartheta)=\mathfrak{g}\otimes \mathfrak{k}_n\otimes\cdots\otimes\mathfrak{k}_1$, where $\psi=|\psi\rangle\otimes |\xi_{n}\rangle\otimes\cdots\otimes|\xi_{1}\rangle$ and $U^{t_i}_{t_{i-1}}\in\mathcal{B}(\mathfrak{g}\otimes\mathfrak{k}_n\otimes\cdots\otimes\mathfrak{k}_i)$.

One might think that the intuitive approach to this measurement dynamics is to consider only a Hilbert space $\mathfrak{g}=\mathfrak{h}\otimes\mathfrak{k}$ of an object coupled to a single copy of the instrument,  and consequently consider an evolution in $\mathcal{B}(\mathfrak{g})$ of the form $U_{t_n}^t E_{k_n} U_{t_{n-1}}^{t_n} E_{k_{n-1}}\cdots E_{k_1} U_0^{t_1}\psi=U_{t_n}^t(V(\vartheta)|\psi\rangle\otimes |\xi_{k_n}\rangle)$, where $U_r^s$ has the Hamiltonian form $e^{-\mathrm{i}H(s-r)}$ and $\psi=|\psi\rangle\otimes|\xi_0\rangle$ and $E_{k_{n}}=1_\mathfrak{h}\otimes P_{k_n}$. However, such evolution gives a very specific sequence of contractions on $\mathfrak{h}$:
\begin{equation}\label{gm2}
    V(\vartheta)|\psi\rangle=V_{k_nk_{n-1}}(t^n_{n-1})\cdots V_{k_10}(t_1)|\psi\rangle
\end{equation}
where $t^n_{n-1}:=t_n-t_{n-1}$ and $V_{jk}(r)=(1_\mathfrak{h}\otimes\langle\xi_j|)e^{-\mathrm{i}Hr}(1_\mathfrak{h}\otimes|\xi_k\rangle)$,
and the dilation of this dynamics that results in a predictable structure, necessarily given by \eqref{gm1}, does not appear to admit the specific contractions in \eqref{gm2}, as any attempt to identify the unitary operators $U_{t_{i-1}}^{t_{i}}$ in \eqref{gm1} with Hamiltonian evolution fails.
Nonetheless, there is a way to rectify this issue  if one considers  the Hilbert space $\mathcal{H}(\vartheta)$ as above with $\mathfrak{k}_i=\mathfrak{k}$ and $\mathfrak{g}=\mathfrak{h}\otimes\mathfrak{k}$, and consider measurement dynamics \eqref{gm2} resulting from unitary operators $U^{t_i}_{t_{i-1}}$ in \eqref{gm1} that \emph{do not have Hamiltonian form}. Before establishing this the following definition will be useful.
\begin{Definition}
    Let $\mathcal{H}(\vartheta)=\mathfrak{g}\otimes \mathfrak{k}_n\otimes\cdots\otimes\mathfrak{k}_1$, then the (semi-tensor) chronological $\odot$-product  is defined as an operator product on $\mathcal{B}(\mathfrak{g})$ and a tensor product between any $\mathcal{B}(\mathfrak{k}_i)$ and $\mathcal{B}(\mathfrak{k}_j)$. For example, if $X_i\otimes Y_i\neg \mathcal{B}(\mathfrak{g}\otimes\mathfrak{k}_i)$ then
 \begin{equation}
 (X_i\otimes Y_i)\odot (X_j\otimes Y_j)=(X_iX_j)\otimes Y_i\otimes Y_j
 \end{equation}
 where the chronology of the action of $X_i$ and $X_j$ is preserved on $\mathfrak{g}$.
 \end{Definition}
Now we can define the wave-function dynamics with respect to a Hamiltonian $\hbar H\neg\mathcal{B}(\mathfrak{g})$ as
\begin{equation}\label{gm3}
   \psi_t(\vartheta)= (e^{-\mathrm{i}Ht}\otimes I(\vartheta)) S(t_n)\odot \cdots \odot S(t_1) \psi
\end{equation}
where $\psi=|\psi\rangle\otimes|\xi_0\rangle \otimes |0\rangle^{\otimes n}$ and $I(\vartheta)=1_n\otimes\cdots\otimes 1_1$ is the apparatus identity operator, and
\begin{equation}
    S(t_i)=(e^{\mathrm{i}Ht_i}\otimes 1_i)S(e^{-\mathrm{i}Ht_i}\otimes 1_i)\in\mathcal{B}(\mathfrak{g}\otimes \mathfrak{k}_i)
\end{equation}
is the `proper' object-apparatus interaction operator. It is both unitary and self-adjoint and, in the standard basis, has the form
\begin{equation}
S=\left(
    \begin{array}{ccc}
      E_0  & E_1 & \ldots \\
      E_1 & 1_\mathfrak{g}-E_1& 0\\
      \vdots  & 0 & \ddots
    \end{array}
  \right),
\end{equation}
which operates in $\mathcal{B}(\mathfrak{k}\otimes\mathfrak{k}_i)$ at the instant of time $t_i$, where $\mathfrak{k}\subset \mathfrak{g}=\mathfrak{h}\otimes\mathfrak{k}$.
The dynamics described by \eqref{gm3} satisfies all of the requirements of observation. First of all it is consistent with \eqref{gm1} if one defines $U^{t_i}_{t_{i-1}}=e^{-\mathrm{i}Ht_i}S(t_i)e^{\mathrm{i}Ht_{i-1}}$, with $S(t)=1$ if there is no measurement at the present time $t$ otherwise we'd replace $t$ with $t_{n+1}$. Secondly, the action of the projector $1_\mathfrak{g}\otimes P_{k_n}\otimes\cdots\otimes P_{k_1}$ on \eqref{gm3} gives, precisely, the dynamics in $\mathfrak{g}$ described by \eqref{gm2}.

At this stage it looks like all is well. However, on closer inspection there is some conceptual ambiguity. On the one hand it looks like the quantum system under observation is $\mathfrak{g}=\mathfrak{h}\otimes\mathfrak{k}$ and the apparatus is the remaining $\mathfrak{k}^{\otimes n}$, but on the other hand, we had liked to interpret the factor $\mathfrak{k}\subset\mathfrak{g}$ as the apparatus, which it is, in part. After all,  the copy of the apparatus contained within $\mathfrak{g}$ encodes all of the interaction with the quantum system $\mathfrak{h}$ but the \emph{actual observations} are encoded in $\mathfrak{k}^{\otimes n}$. In fact,  the interaction operators $S(t_i)$ couple a $\mathfrak{k}_i$,  from the observable part of the apparatus $\mathfrak{k}^{\otimes n}$, with the non-observable part of the apparatus $\mathfrak{k}\subset\mathfrak{g}$. Moreover, the unitary  interaction $S$ is not of the Hamiltonian form, it is a \emph{spontaneous interaction}.

The issue now is whether or not we need to have a copy of $\mathfrak{k}$ entangling with the quantum system $\mathfrak{h}$ in this Hamiltonian manner, i.e. $e^{-\mathrm{i}Ht}$ on $\mathfrak{g}=\mathfrak{h}\otimes\mathfrak{k}$, because we could instead just have $\mathfrak{g}=\mathfrak{h}$ and consider the object-apparatus coupling arising from spontaneous interactions $S$ on $\mathfrak{h}\otimes\mathfrak{k}_i$, where $\mathfrak{k_i}\subset \mathfrak{k}^{\otimes n}$. Another way of putting this is: why do we want \eqref{gm2} to be satisfied? As motivation for asking this question
consider the case where an atom $\mathfrak{h}$ is represented as a two level system and its decay is observed via its coupling to Schr\"odinger's cat. In this case the cat is serving as the instrument of the apparatus and if we have interaction given by \eqref{gm3} with $\mathfrak{g}=\mathfrak{h}\otimes\mathfrak{k}$ then it is not clear how to construct a Hamiltonian that will \emph{not} allow the cat to die more than once. However, in Section \ref{4} it is shown that we can take $\mathfrak{g}=\mathfrak{h}$ and have a complete quantum theory of Schr\"odinger's cat that is consistent with reality.  First, we must complete our construction of the apparatus. So far, we have only considered the case of $n$ observations at fixed times $t_i$ in a chain $\vartheta$, but we'd like to consider arbitrary instants of observation.

\begin{Definition}
    A trajectory is a time-ordered sequence of observations  up to a time $t$, represented by a projector $\Pi_t(\vartheta)$ as
    \begin{equation}
    \Pi_t(\vartheta)=I(\vartheta_{t})\otimes P_{k_n}(t_n)\otimes\cdots\otimes P_{k_1}(t_1)
    \end{equation}
    where $\vartheta=\vartheta_{[t}\sqcup\vartheta^t$ is the decomposition of the chain $\vartheta$ into future $\vartheta_{[t}$ and past $\vartheta^t$, and $|\vartheta^t|=n$.
\end{Definition}

 \begin{Theorem}
An apparatus capable of measurement outcomes $k\in\Omega$ at random times $\vartheta=\{t_n>\cdots>t_1\}$ in a connected interval $T\subseteq\mathbb{R}$ for a random number of observations $|\vartheta|=n$ is represented on the Guichardet second quantisation $\mathcal{F}=\Gamma(\mathcal{K})$ of the Hilbert space $\mathcal{K}=\mathfrak{k}\otimes L^2(T)$, where $\mathfrak{k}=L^2(\Omega)$. Moreover, $\mathcal{F}$ is a continuous tensor product space, admitting the future-past decomposition $\mathcal{F}=\mathcal{F}_{[t}\otimes\mathcal{F}^t$ at any time $t\in T$.
 \end{Theorem}
 \begin{proof}
 For any given chain $\vartheta\subset T$ Lemma \ref{lem2} implies that the apparatus Hilbert space has the form $\mathcal{F}(\vartheta):=\mathfrak{k}^{\otimes |\vartheta|}$, where $|\vartheta|$ is the cardinality of $\vartheta$. For any $\vartheta=\vartheta_n$ of fixed cardinality $|\vartheta_n|=n$ each Hilbert space $\mathcal{F}(\vartheta_n)$ may be regarded as a fibre of a trivial Hilbert bundle $\mathcal{F}_n$ over an $n$-simplex $\mathcal{T}_n$ of the chains $\vartheta_n$. To see this, note that any function $\Phi_n\in\mathcal{F}_n$ has evaluations $\langle\vartheta_n|\Phi_n=\Phi_n(\vartheta_n)$, given with respect to the generalised Dirac-type eigenvectors $|\vartheta\rangle=|t_n\rangle\otimes\cdots\otimes|t_1\rangle$ with $\langle t|s\rangle=\delta(t-s)$. Moreover, the requirement for $\mathcal{F}_n$ to also be a Hilbert space  is that it is eqiupped with an inner-product defining the squared norm
 \begin{equation}\label{nmeas}
 \Phi_n^\dag\Phi_n^{}=\int _{\mathcal{T}_n}\|\Phi_n(\vartheta_n)\|^2\mathrm{d}\vartheta_n:=\underset{t_1<\cdots <t_n}{\int\cdots\int}  \|\Phi_n(t_1,\ldots,t_n)\|^2\mathrm{d}t_1\cdots\mathrm{d}t_n<\infty
 \end{equation}
 where $\|\Phi_n(t_1,\ldots,t_n)\|$ is the norm induced by the inner-product for $\mathfrak{k}^{\otimes n}$,   and we see that $\mathcal{F}_n\cong \mathfrak{k}^{\otimes n}\otimes L^2(\mathcal{T}_n)$. Note that if the measurement outcomes $\Omega$ were regarded to vary at different times then the Hilbert bundle would be non-trivial. Finally, the requirement for $n$ to take \emph{any} value in the disjoint union $\mathbb{N}_0:=\mathbb{N}\sqcup\{0\}$ is that the whole Hilbert space representing the apparatus is the \emph{Guichardet-Fock space} \cite{GUI} defined as
 \begin{equation}\label{gfs}
\mathcal{F}=\bigoplus_{n\in\mathbb{N}_0}\mathcal{F}_n,\quad\mathcal{F}_0:=\mathbb{C}
 \end{equation}
 which may be identified with a restriction of the full Fock space over $\mathcal{K}:=\mathfrak{k}\otimes L^2(T)$ to a simplex of chains having the inner product
 \begin{equation}\label{nmeas2}
 \Phi^\dag\Phi=\int _\mathcal{T}\|\Phi(\vartheta)\|^2\mathrm{d}\vartheta:=\sum_{n\in\mathbb{N}_0}\int _{\mathcal{T}_n}\|\Phi_n(\vartheta_n)\|^2\mathrm{d}\vartheta_n<\infty
 \end{equation}
 for any $\Phi=\oplus_n\Phi_n\in\mathcal{F}$, where $\mathcal{T}=\cup\mathcal{T}_n$ is the simplex of chains of any length, including the empty chain $\varnothing$, $|\varnothing|=0$.

 The continuous tensor product structure arises from the arbitrary decomposability $\mathcal{F}=\mathcal{F}^t\otimes\mathcal{F}_{[t}$ at any time $t\in T$. This is proved as follows. First identify $T=[a,b)$, $a<b$, then consider the disjunction $T=[a,t)\sqcup[t,b):=T^t\sqcup T_{[t}$. This corresponds to the disjunction of chains $\vartheta=\vartheta^t\sqcup\vartheta_{[t}$. Now, there is the so-called sum-integral formula \cite{CHAO} for functions of chains, for $f\in L^1(\mathcal{T}\times\mathcal{T})$ it is
 \begin{equation}\label{suin}
 \iint f(\varkappa,\varsigma)\mathrm{d}\varkappa\mathrm{d}\varsigma=\int\sum_{\varkappa\sqcup\varsigma=\vartheta}f(\varkappa,\varsigma)\mathrm{d}\vartheta. \end{equation}
 As a special case one can consider functions $f\in L^1(\mathcal{T}^t\times\mathcal{T}_{[t})$ and embed them into $L^1(\mathcal{T}\times\mathcal{T})$ as $f(\varkappa,\varsigma):=0$ if $\varkappa\cap T_{[t}\neq\varnothing$ or $\varsigma\cap T^t\neq\varnothing$. This gives \eqref{suin} in the form
 \begin{equation}\label{ion}
 \iint f(\vartheta^t,\vartheta_{[t})\mathrm{d}\vartheta^t\mathrm{d}\vartheta_{[t}=\int f(\vartheta^t,\vartheta\setminus\vartheta^t)\mathrm{d}\vartheta.
 \end{equation}
  Now, given any function $\phi \in\mathcal{F}^t\otimes\mathcal{F}_{[t}$ we can establish one-to-one correspondence with a function $\Phi\in\mathcal{F}$ simply as
 \begin{equation}
\Phi(\vartheta)=\phi(\vartheta^t,\vartheta\setminus\vartheta^t)
 \end{equation}
 and \eqref{ion} establishes the equivalence of norms $\|\Phi\|_\mathcal{F}=\|\phi\|_{\mathcal{F}^t\otimes\mathcal{F}_{[t}}$.
 In this way we can define the Guichardet-Fock second quantisation functor $\Gamma:\mathcal{K}\rightarrow\mathcal{F}$ having the property $\Gamma(\mathcal{K}^t\oplus\mathcal{K}_{[t})=\Gamma(\mathcal{K}^t)\otimes\Gamma(\mathcal{K}_{[t})$.
 \end{proof}
 To understand the full picture for the apparatus as represented by $\mathcal{F}$ we shall consider the prepared instrument state-vector $|\xi\rangle$ and compose this with a function $\sqrt{\nu}\in L^2(T)$ describing how the measurements of the quantum system are localised in time. Now we have $\xi_\nu=|\xi\rangle\otimes\sqrt{\nu}\in\mathcal{K}$ and the state of the apparatus is represented by the normalised vector $\Phi=\xi_\nu^\otimes e^{-\frac{1}{2}\int\nu\mathrm{d}t}\in\mathcal{F}$, which may also be given by the action of the unitary Weyl-operator $W(\xi_\nu)$ on the vacuum vector $\delta_\varnothing:=0^\otimes$. Here, $\langle\vartheta_n|\xi_\nu^\otimes=\xi_\nu(t_n)\otimes\cdots\otimes\xi_\nu(t_1)$.  We shall now see that $\nu$ is the \emph{observation frequency}, describing the rate at which the quantum system is being observed, the intensity of a quantum Poisson process.
 \begin{Corollary}
 The Hilbert space $\mathcal{H}=\mathfrak{h}\otimes\mathcal{F}$, where $\mathcal{F}=\Gamma(\mathfrak{k}\otimes L^2(T))$, of a quantum system under observation is also the representing Hilbert space for the quantum stochastic dynamics of a quantum system $\mathfrak{h}$  subjected to quantum noise with degrees of freedom determined by the dimensions of $\mathfrak{k}$.
 \end{Corollary}
 \begin{proof}
     The representing Hilbert space of quantum stochastic differential equations is well-known and originally constructed in \cite{HP}.
 \end{proof}

 The structure of Quantum Stochastic Calculus (QSC) began to emerge in the late 1960's and 1970's when Davies, Lindblad \emph{et al.} established a description of quantum dynamical semigroups \cite{L,D,G}. But the Lindblad equation only described an average of a stochastic dynamics of a quantum system in $\mathcal{B}(\mathfrak{h})$, and it was not until  Hudson and Parthasarathy's discovery of the structure of quantum stochastic calculus (generalising It\^o's classical theory) that the Lindblad equation could be properly realised as the marginal dynamics of a unitary evolution on a pure-state in $\mathfrak{h}\otimes\mathcal{F}$.

 At this stage one might object to the underlying assumption in this article that the (indirect) observations of the quantum system $\mathfrak{h}$ are `spontaneous'. That is, they are represented by projections defined at  singular instants of time $t_i\in\vartheta$. However, this QSC formalism allows for continuous observation too \cite{BS,CM}.  The fundamental example of such continuous observation corresponds to quantum Brownian motion, and such diffusive type observation was also studied in \cite{BLP,Di,GG}.
 Here we are restricting our attention to a discrete sequence of observations of a quantum system (which corresponds to a quantum Poisson process) just to see how the apparatus Hilbert space is built. In fact, we could have considered continuous observation from the start by considering a differential version of the conditional expectation, but this involves a generalisation of the Hilbert space theory to a Krein space theory \cite{CHAO,MFB}.

 The dynamics given by \eqref{gm3} is fully realised on $\mathcal{H}=\mathfrak{g}\otimes\mathcal{F}$ as $\psi_t=U_0^t\psi$, with $\min T=0$, and described by the quantum stochastic differential equation (QSDE)
 \begin{equation}\label{34}
\mathrm{d}\psi_t(\vartheta)+\mathrm{i}H\psi_t(\vartheta)\mathrm{d}t=L(t)\odot\psi_t(\vartheta\setminus t)\mathrm{d}n_t(\vartheta),
 \end{equation}
 where $H\equiv H\otimes I(\vartheta)$, $L(t)=(S-I)\xi_\nu(t)$ and $n_t(\vartheta)=|\vartheta^t|I(\vartheta)$ is fundamental  counting process with $\mathrm{d}n_t(\vartheta)=1$ if $t\in\vartheta$ and otherwise zero, and it has expectation $\nu t$ determined by the pure-state $\mathbb{E}_t$ given by $\psi_t=U_0^t(|\psi\rangle\otimes\Phi)$. The only difference between this $\psi_t(\vartheta)$ and that in \eqref{gm3} is the  $n$ extra factors of $\sqrt{\nu}$.

This generalises the usual Schr\"odinger dynamics and is just a special case of more general QSDEs which can be written in Belavkin's concise form \cite{CHAO}, also resembling the formalism in \cite{LM}, as
 \begin{equation}
\mathrm{d}U_0^t=L^\mu_\kappa(t)\mathrm{d}\Lambda^\kappa_\mu(t) U_0^t
 \end{equation}
 with respect to the quantum stochastic differential increments of creation $\mathrm{d}\Lambda^+_k$, annihilation $\mathrm{d}\Lambda^k_-$, preservation $\mathrm{d}\Lambda^+_-=\mathrm{d}t$ and scattering $\mathrm{d}\Lambda^i_k=\mathrm{d}N^i_k$, satisfying the Hudson-Parthasarathy multiplication table, where $i,k\in\Omega$, $\kappa\in\{\Omega,+\}$ and $\mu\in\{-,\Omega\}$ and $\kappa$ and $\mu$ are summed over. There is a rich structure encoded in the Belavkin Representation of QSC. To recover \eqref{34} just take $L^-_k=0=L^k_+$, $L^-_+(t)=-\mathrm{i}H$ and $ L^k_i(t)\mathrm{d}N^i_k(t)=(S^k_i-\delta^k_i)\mathrm{d}N^i_k(t)=(S-I)\mathrm{d}n_t$.

 The unitary quantum stochastic evolution operator $U_0^t$ is adapted to $\mathcal{F}_{[t}\otimes\mathcal{H}^t$, which means that it acts as the identity on the apparatus' future $\mathcal{F}_{[t}$. Further, it forms a cocycle, such that
 \begin{equation}
 U_r^t=U_s^t U_r^s.
 \end{equation}
 In fact, for the adapted $\odot$-product dynamics of \eqref{34} we have
 \begin{equation}
 U_r^t\equiv I^{}_{[t}\otimes U_{[s}^t\odot U_{[r}^s\otimes I^r,
 \end{equation}
 where lower indices $[x$ mean "from, and including, $x$" and raised indices $x$ mean "up to, and excluding, $x$".
 So we now see that in the case of countable observations the interaction dynamics, coupling an apparatus to a quantum system, is given by $U_0^t$, a unitary Poisson process with rate $\nu$. Such Poisson processes admit the diagonal form
 \begin{equation}
 U_0^t=e^{-\mathrm{i}Ht}\Big(I_{[t}\otimes\int_{\mathcal{T}^t} |\vartheta\rangle S^\odot(\vartheta)\langle \vartheta|\mathrm{d}\vartheta\Big),
 \end{equation}
 where $S^\odot(\vartheta)$ is the $\odot$-product of the interaction operators $S(t_i)=e^{\mathrm{i}Ht_i}Se^{-\mathrm{i}Ht_i}$, where $t_i\in\vartheta$.
In the case of more general (continuous) observations the unitary operator $U_0^t$ does not have this diagonal form, but is still represented on $\mathfrak{h}\otimes\mathcal{F}$.

The Lindblad dynamics of a quantum system is simply given by $\varrho(t)=\texttt{Tr}_\mathcal{F}[\psi_t^{}\psi_t^\dag]$, and general Lindblad equations are given as the marginal dynamics of quantum unitary Brownian motions. However, such unitary Poisson processes that we restrict our attention to here only give us a special class of Lindblad equations of the form
\begin{equation}\label{39}
 \dot\varrho
 =\mathrm{i}[\varrho,H_\text{eff}]+\widetilde{L}(\varrho\otimes 1_\mathfrak{k})\widetilde{L}^\dag
 -\tfrac{1}{2}\{L^\dag L,\varrho\}=\mathrm{i}[\varrho,H]+\nu(\varphi(\varrho)-\varrho)
\end{equation}
where the time dependence has been dropped from the notation for simplicity. Here, $H_{\text{eff}}=H+V_\text{eff}(t)$ and $V_{\text{eff}}(t)=\mathrm{i}\xi_\nu^\dag(t)(S-S^\dag)\xi_\nu(t)/2$ is an effective potential arising from the interaction,  $L(t)=(S-I)\xi_\nu(t)\in \mathcal{B}(\mathfrak{h})\otimes\mathfrak{k}$ is a column-vector of operators and $\widetilde{L}$ is its partial transposed row (transpose in $\mathfrak{k}$ only), $\varphi$ is a unital CP map and $\{A,B\}=AB+BA$ is the anti-commutator. The second equality in \eqref{39} follows from direct computation from the definition of $L$ and the identity $\xi_\nu^\dag S\xi_\nu=\xi_\nu^\intercal \widetilde{S}\bar\xi_\nu$, noting that $\nu\varphi(\varrho):=\xi_\nu^\intercal \widetilde{S}(\varrho\otimes 1_\mathfrak{k})\widetilde{S}^\dag\bar\xi_\nu$.

The final part of the full observation dynamics outlined here is the conditioning of quantum properties $X\neg\mathcal{B}(\mathfrak{h})$ with respect to the algebra of all trajectories up to a time $t$, denoted by $\mathcal{C}_t\subset\mathcal{B}(\mathcal{F})$. By construction $\mathcal{B}(\mathfrak{h})\subset\mathcal{A}_t=\mathcal{C}'_t$ and so $X$ is $\mathbb{E}_t|\mathcal{C}_s$-predictable, where $s\leq t$ and $\mathbb{E}_t$ is a pure state given by $\psi_t=U_0^t\psi$. The conditional expectation that comes from this observation is $\epsilon_t(X)=\mathbb{E}_t[X|\mathcal{C}_t]$.
We shall conclude this section with a final definition and  by also noting that
the object-apparatus interaction may be described in a far more general way than the conventional Hamiltonian approach. Namely, by quantum stochastic  cocycles. Further, this also means that we may understand the action of an apparatus on a quantum system as noise.

 \begin{Definition}\label{def6}
 A quantum filter is a conditional expectation $\epsilon_t:\mathcal{A}_t\rightarrow\mathcal{C}_t$ where $\mathcal{C}_t\subset\mathcal{B}(\mathcal{F})$ is the abelian von Neumann algebra generated by all possible compatible trajectories up to a time $t\in T$  defining an adapted filtration such that $\mathcal{C}_t=I_{[t}\otimes\mathcal{C}^t$ and $\mathcal{C}_s\subset\mathcal{C}_t$ for all $s<t$. The algebra $\mathcal{A}_t$ is, at most, the commutant of $\mathcal{C}_t$ in $\mathcal{B}(\mathfrak{h}\otimes\mathcal{F})$ and $\mathcal{B}(\mathfrak{h})\subset\mathcal{A}_t$.
 \end{Definition}
  Quantum Filtering was developed by Belavkin in the 1980's  and its importance as the solution to the quantum measurement problem appears to be largely unknown in the physics community.
In the next section it shall be shown how to construct a quantum filter for the Schr\"odinger cat problem.

\section{Schr\"odinger's Cat is a Quantum Filter}\label{4}
Here we consider the usual setup: a living, observable, cat is coupled to a two-level quantum system called \emph{atom} in such a way that if this atom decays then the cat dies. It should be clear that the possible observations $\Omega$ we can make, of the cat, are represented as a compatible system  of orthoprojectors $\{P_k:k\in\Omega\}$ with which the atom may be conditioned. The way in which the atom and the cat are coupled is given by choosing, for each type of cat observation, a corresponding operation on the atom.

Let $\mathfrak{h}=\mathbb{C}^2$ be the atom Hilbert space consisting of state-vectors $|\psi\rangle=\alpha|g\rangle+\beta|e\rangle$, denoting `ground' and `excited' states. This atom is also considered to evolve according to some Hamiltonian $\hbar H$. Meanwhile, we shall consider the cat Hilbert space $\mathfrak{k}$, the dimension of which is yet to be determined and not necessarily 2. The cat Hilbert space is the representing space for the projections corresponding to the observations of the cat. The first of these is the most important and is the observation of a \emph{dying} cat. We shall call this $P_1$ and the corresponding action on the atom must be the annihilation of the excited state: $J:=|g\rangle\langle e|$. The next observation that we could make is one in which the cat is alive, $P_2$. If the cat is alive then presumably nothing has happened to the atom, so we'd assume that the action on the atom is given by the identity operator $1_\mathfrak{h}$. Finally, we consider a projector $P_0$ resulting in  $\mathfrak{k}=\mathbb{C}^3$ and  enabling us to construct the unitary atom-cat interaction operator
\begin{equation}\label{S}
S=\left(
    \begin{array}{ccc}
      \sqrt{p}JJ^{\dag}  & \sqrt{p}J^{\dag} & \sqrt{q}1_\mathfrak{h} \\
      \sqrt{p}J & J^{\dag}J+{q}JJ^\dag& -\sqrt{qp}J \\
      \sqrt{q}1_\mathfrak{h}  & -\sqrt{qp}J^\dag & pJ^\dag J-\sqrt{p}JJ^\dag
    \end{array}
  \right)
\end{equation}
with respect to the standard basis,
where $J^\dag$ is the usual Hermitian adjoint of $J$ induced by the innerproduct on $\mathfrak{h}\times \mathfrak{h} $ and $0\leq p=1-q\leq1$.
The form of $S$ has been defined by the choice of instrument basis, $P_k=|k\rangle\langle k|$, and the important operator is  the isometry $\widetilde{V}:=S|0\rangle$ which defines a quantum channel
\begin{equation}\label{qch}
    \varphi(\varrho) = {V} (\varrho\otimes I){V}^\dag=p |g\rangle\langle g|+q\varrho
\end{equation}
on any atomic density matrix $\varrho$. Here, the notation $\widetilde{V}\in\mathcal{B}(\mathfrak{h})\otimes\mathfrak{k}$ denotes the \emph{partial transpose} (transpose of $\mathfrak{k}$ only) of $V\in\mathcal{B}(\mathfrak{h})\otimes\mathfrak{k}^\intercal$.

As we saw in the previous section, such an operation $S$ corresponds to the atom-cat coupling at any instant of time, $t$. Each instant of the cat is the instrument of our apparatus here, and this is not just considered at a single instant of time but over some continuous interval of time $T\subset\mathbb{R}$. On such a time interval we can consider any number of `spontaneous' observations of the cat, indexed by chains $\vartheta=\{t_n >\cdots>t_2>t_1 \}\subset T$, and such a sequence of observations is represented on the Hilbert space $\mathfrak{k}^{\otimes n}$, where $n=|\vartheta|$. This is an $n$-point evaluation of the continuous tensor product space $\mathcal{F}$, which is also the Guichardet second quantisation $\Gamma(\mathcal{K})$ of $\mathcal{K}=\mathfrak{k}\otimes L^2(T)$.

Thus, in order to observe the quantum system by virtue of its coupling to the observable cat apparatus we have the following complete atom-cat interaction dynamics. The representing Hilbert space is $\mathcal{H}=\mathfrak{h}\otimes\mathcal{F}$, where $\mathcal{F}$ admits the future-past (causal) decomposition $\mathcal{F}_{[t}\otimes\mathcal{F}^t$. The initial atomic state-vector is $|\psi\rangle$ and the apparatus is defined in the coherent state vector $\Phi :=\xi_\nu^\otimes e^{-\frac{1}{2}\int \nu \mathrm{d}t}$, where $\xi_\nu:=|0\rangle\otimes\sqrt\nu\in\mathcal{K}$ and $\nu$ is the frequency of cat observations.  The atom-cat coupling may itself be understood as the conditioning of the object with respect to the apparatus. We shall work in the \emph{interaction picture} where the interaction dynamics is given by the (unitary) quantum stochastic cocycle $U_0^t=S^\odot_t$, taking $\min T=0$, so that $\psi_t=U_0^t(|\psi\rangle\otimes\Phi)$ has the evaluations
\begin{equation}\label{ev}
\psi_t(\vartheta)=\Phi(\vartheta_{[t})\otimes S^\odot(\vartheta^t)(|\psi\rangle\otimes\Phi(\vartheta^t))=\Phi(\vartheta_{[t})\otimes(\sqrt{\nu}\widetilde{V})^\odot(\vartheta^t)|\psi\rangle
\end{equation}
resolving the QSDE
\begin{equation}\label{dyn}
\mathrm{d}\psi_t(\vartheta)=L(t)\odot\psi_t(\vartheta\setminus t)\mathrm{d}n_t(\vartheta),
\end{equation}
where $L(t)=\sqrt{\nu}(t)(S(t)-I)|0\rangle\in\mathcal{B}(\mathfrak{h})\otimes\mathfrak{k}$ (with $S(t)$, not $S$, in contrast to \eqref{34}) and $n_t(\vartheta)=|\vartheta^t|I(\vartheta)$ is the input counting process with  mean $\psi_t^\dag (1_\mathfrak{h}\otimes n_t)\psi_t^{}=\Phi^\dag n_t\Phi=\nu t$ and intensity $\nu$, having increments $\mathrm{d}n_t(\vartheta)=1$ if $t\in\vartheta$ and otherwise zero.

In the interaction picture any predictable atomic property $X\in\mathcal{B}(\mathfrak{h})$ will evolve according to the free Heisenberg dynamics given by the object Hamiltonian $\hbar H$.   Also notice that the cat, or indeed each copy of the cat in a continuum of cats  over $T$, is coupled with the atom up to a time $t$ via the interaction $U_0^t$, so that the state of the whole system is an entangled quantum pure-state. Nonetheless, we'll see that the possible observations of the cat are the trajectories $\Pi_t\in\mathcal{C}_t$; this is consistent with our every-day experience in experiments.

Before we proceed with the observed trajectories and the filtering theory let's first review the marginal dynamics of the object. This is of course given by $\varrho(t)=\texttt{Tr}_\mathcal{F}[\psi_t^{}\psi_t^\dag]$. To simplify this we shall consider $\sqrt{\nu}$ to be constant on $T$  then for $t\in T$ the marginal density has the explicit form
\begin{eqnarray}\nonumber
    \varrho(t)
    =\widetilde{\Phi}\widetilde{U}_0^t (\rho(0)\otimes I^\otimes)\widetilde{U}_0^{t\dag}\widetilde{\Phi}^\dag
    &=&\sum_{n=0}^\infty\int_{\mathcal{T}^t}V(\vartheta_n)\left(\varrho(0)\otimes I^{\otimes n} \right)V^\dag(\vartheta_n) \nu^n e^{-\nu t}\mathrm{d}\vartheta_n,\\
    &=&\sum_{n=0}^\infty V^{\dag\odot n\dag}\left(\varrho(0)\otimes I^{\otimes n} \right)V^{\dag\odot n} \frac{(\nu t)^n}{n!} e^{-\nu t}\label{mar}
\end{eqnarray}
where $\varrho(0)=|\psi\rangle\langle\psi|$, $\int\mathrm{d}\vartheta_n=\idotsint \mathrm{d}t_1\cdots\mathrm{d}t_n|0<t_1<\cdots<t_n<t$ is multiple integration over the simplex of  chains, and $V^\dag(\vartheta_n):=V^\dag(t_1)\odot\cdots\odot V^\dag(t_n)$ which, under the assumption that $H=\varepsilon_g|g\rangle\langle g|+\varepsilon_e|e\rangle\langle e|$, can be replaced with $V^{\dag\odot n}$ not depending on $t$. This is because in the interaction picture $V(t)=e^{\mathrm{i}Ht}V(e^{-\mathrm{i}Ht}\otimes 1_\mathfrak{k})=(\sqrt{p}JJ^\dag,\sqrt{p}J(t),\sqrt{q}1_\mathfrak{h})$  can be factorised  as $V(1_\mathfrak{h}\otimes(P_0+u(t)P_1+P_2))$, where $u(t)=e^{-i\varepsilon t}$ with $\varepsilon=\varepsilon_e-\varepsilon_g$, so then we see $V(t)(\varrho(t)\otimes I)V(t)^\dag=V(\varrho(t)\otimes I)V^\dag$, with $V$ not depending on $t$.
In sight of \eqref{qch} we can further simplify \eqref{mar} to
\begin{eqnarray}\nonumber
    \varrho(t)=\sum_{n=0}^\infty \varphi^n(\varrho(0)) \frac{(\nu t)^n}{n!} e^{-\nu t}&=&
    \sum_{n=0}^\infty \left(\frac{p(1-q^{n})}{1-q}|g\rangle\langle g|+q^n \varrho(0)\right)\frac{(\nu t)^n}{n!} e^{-\nu t}\\
    &=&|\psi\rangle\langle\psi|e^{-p\nu t}+|g\rangle\langle g|(1-e^{-p\nu t}).   \label{mar2}
\end{eqnarray}
\begin{Remark}
So we see that the expected behaviour of the atom, the 2-level quantum system, is that it \emph{continually} decays from its initial state into its ground state simply as a result of being coupled to the cat.
\end{Remark}
Moreover, notice that this dynamics may also be given by the Lindblad equation
\begin{equation}
    \partial_t\varrho(t)=\widetilde{L}(\varrho(t)\otimes I)\widetilde{L}^\dag-\tfrac{1}{2}\{L^\dag L,\varrho(t)\}
\end{equation}
where  ${L}=\sqrt{\nu}({S}-I)|0\rangle=\sqrt{\nu}(\widetilde{V}-|0\rangle)$. Following the discussion at the end of Section \ref{3} it is notable that from the atom's `point of view', the cat (which serves as the apparatus) is a field of quantum noise. That is, a system under observation appears to be the same as a system in a field of noise.

\subsection{Observation and Quantum Filtering} Now that the specific mathematical framework has been set up to describe the coupling of an apparatus, the cat, to a quantum object, the atom, we are now in a position to investigate the form of an actual observation, a trajectory. That is an ordered string of data produced in any experiment that tells us something about the object under indirect observation by virtue of its coupling to the apparatus. In contrast to some beliefs about observing atoms, these observations are recognised as a property of the apparatus \emph{not} atomic properties, but we can \emph{infer} information about the atom by virtue of its coupling (interaction) with the apparatus.

One of the purposes of this article is to demonstrate that a quantum property $X\neg\mathcal{B}(\mathfrak{h})$ is not actually observable. Instead it is only $\Omega$-observable, which means it is a \emph{predictable quantum property} which can  be estimated  within the context of an apparatus using the conditional expectation, which is the best estimate (in a least-squares sense) given such an apparatus \cite{LUC}. In Definition \ref{def6} we saw that the apparatus is described by the adapted abelian von Neumann algebra $\mathcal{C}_t$, which is also a filtration and represented on $\mathcal{F}$.  The conditional expectation
  of $X$ with respect to $\mathcal{C}_t$ is denoted $\epsilon_t(X)\equiv\mathbb{E}_t[X|\mathcal{C}_t]$ and it is called the \emph{quantum filter}. As we shall see, the quantum filter is a projection of the quantum system onto the apparatus algebra
  and describes the possible observations that can arise in the context of this apparatus. This projection is given by $E^\dag_t=E_t^\dag E_t=E_t\in\mathcal{A}_t$ (recalling from Section \ref{3} that $\mathcal{A}_t$ is the commutant of $\mathcal{C}_t$ in $\mathcal{B}(\mathcal{H})$) such that $E_t(X\otimes I)E_t=(1_\mathfrak{h}\otimes\epsilon_t(X))E_t$, where $E_t:=\Psi_t^{}\Psi_t^\dag$ is itself defined by the fundamental partial isometry $\Psi_t$ called the \emph{filtering wave-function} \cite{FWF} which may be regarded as the means by which our observations arise, as discussed at the end of Section \ref{1}.

The abelian apparatus algebra $\mathcal{C}_t$  is where the actual observations that arise in an experiment live. These observations, and the algebra $\mathcal{C}_t$, are generated by projectors $\Pi_t$ having evaluations
\begin{equation}\label{tra}
    \Pi_t(\vartheta)=I\big(\vartheta_{[t}\big)\otimes P_{k_n}(t_n)\otimes\cdots\otimes P_{k_1}(t_1)
\end{equation}
where $\vartheta^t=\{t_n>\cdots>t_1\}$ and the indices $k_i\in\{0,1,2\}=\Omega$ denote the type of observation made at time $t_i$. Recall that for each type of observation there is, respectively, a corresponding contraction $V_k\in\{\sqrt{p}|g\rangle\langle g|, \sqrt{p}J, \sqrt{q}1_\mathfrak{h}\}$ on $\mathfrak{h}$, so we see that observations of type-1 cannot follow observations of type-0 or type-1; projectors $\Pi_t$ corresponding to such impossible observations are orthogonal to the total system  wave-function: $(1_\mathfrak{h}\otimes\Pi_t)\psi_t=0$. It is therefore convenient to identify a projector $P_t$ such that the subalgebra $P_t\mathcal{C}_t\subset\mathcal{C}_t$ is the algebra of all possible observations. Indeed, $(1_\mathfrak{h}\otimes P_t)\psi_t=\psi_t$.

Notice that the projectors \eqref{tra} are trajectories: they represent a time-ordered sequence of points in the space of all possible observations. In fact, an observation over an interval of time is really a generalisation of the notion of a trajectory.

The complete set of orthoprojectors for observation up to a time $t$ are  as follows. The first is  the adapted vacuum projector $\Pi_t^\varnothing(\vartheta)=I(\vartheta_{[t})\otimes O(\vartheta^t)$  corresponding to `no observation', where $O(\vartheta):=1$ if $\vartheta=\varnothing$ and zero otherwise.  Then there are two other types of observations. The first are
   represented by orthoprojectors of the form
\begin{equation}\label{ob0}
  \Pi^0_t(\vartheta)=I(\vartheta_{[t})\otimes P_0^\otimes(\varkappa)\otimes P_2^\otimes(\omega)
\end{equation}
for any partition $\varkappa\sqcup\omega=\vartheta^t$,
and the second are represented by orthoprojectors of the form
\begin{equation}\label{ob1}
    \Pi^1_t(\vartheta)=I(\vartheta_{[t})\otimes  P_0^\otimes(\varkappa)\otimes P^{}_1(t_i)\otimes P_2^\otimes(\omega)
\end{equation}
for a $t_i\in\vartheta^t$ and any partition $\varkappa\sqcup\omega=\vartheta^t\setminus t_i$, where $\Pi_t^k(\varnothing):=0$ for $k=0,1$. Now, there is an analogue of Newton's binomial formula for product functions over the simplex of chains:
\begin{equation}
\sum_{\varkappa\sqcup\omega=\vartheta}P^\otimes(\varkappa)\otimes Q^\otimes(\omega)=(P+Q)^\otimes(\vartheta),
\end{equation}
so if we sum \eqref{ob0} over all partitions  $\varkappa\sqcup\omega=\vartheta^t$ corresponding to observations of type-0 and type-2, e.g. 00202200020222020202..., then we end up with the projector
\begin{equation}
    P_t^0(\vartheta)= I(\vartheta_{[t})\otimes (P_0+P_2)^\otimes(\vartheta^t)
\end{equation}
if $|\vartheta^t|\neq0$, and we've defined $P_t^0(\varnothing)=0$ for convenience (in order to have a distinct projector $\Pi^\varnothing_t$).
We can proceed in a similar manner with \eqref{ob1}, summing over the partitions of $\varkappa\sqcup\omega=\vartheta^t\setminus t_i$, to get
\begin{equation}
     P_t^1(\vartheta)=\sum_{t_i\in\vartheta^t} I(\vartheta_{[t})\otimes (P_0+P_2)^\otimes(\vartheta^t\setminus t_i)\otimes P_1(t_i).
\end{equation}
Then, defining $P_t^\varnothing=\Pi_t^\varnothing$ for convenience, we see that
\begin{equation}
    P_t=P_t^\varnothing+P_t^0+P_t^1
\end{equation}
is a decomposition of of $P_t$ into three fundamental types of observation given by orthogonal projections.

The projections $P_0$, $P_1$, $P_2$ $\in\mathcal{B}(\mathfrak{k})$, from which the trajectories \eqref{ob0} and \eqref{ob1} are built, may be interpreted as follows. If the cat is observed to be dying then this is represented by $P_1$; this can only happen once due to the nilpotent nature of the operator $J$ on $\mathfrak{h}$. Next, in view of \eqref{S} recall the vector of contraction operators $\widetilde{V}=S|0\rangle$. If $q=1-p=1$, then the atom and cat are not even coupled, so $p$ should be interpreted as an atom-cat coupling probability so that $S(|\psi\rangle\otimes|0\rangle)$ may be thought of as a superposition of `cat coupled to atom' with `cat not coupled to atom'. In that case, an observation of the form $P_2$ can be interpreted as an observation of the cat that infers nothing about the atom. Finally, we have observations of type-0. These are strange in so far as if the cat has died, then $P_0$ represents an observation of a dead cat that leads to the inference of an atomic ground state. On the other hand, if the cat has not died then $P_0$ represents an observation of a living cat that leads to the inference that the atom was always in its ground state and so it will not decay.

To see the mathematical realisation of these interpretations we need only evaluate the expectations of the projectors $P_t^0$ and $P_t^1$; which correspond to the probabilities of the events that they represent. Recall that the initial atomic state was assumed to be $|\psi\rangle=\alpha|g\rangle+\beta|e\rangle$, and that $P_t^k(\varnothing):=0$, $k=0,1$, then we find
\begin{equation}
    \mathbb{E}_t[P_t^0]=\psi_t^\dag(1_\mathfrak{h}\otimes P_t^0 )\psi_t=|\alpha|^2(1-e^{-\nu t})\quad\text{and}\quad
    \mathbb{E}_t[P_t^1]=|\beta|^2(1-e^{-\nu t})
\end{equation}
with $\mathbb{E}_t[P_t^\varnothing]=e^{-\nu t}$ showing that $\mathbb{E}_t[P_t]=1$. Indeed, if $\alpha=0$ the probability that the cat will die will initially be zero but eventually approach 1;  more generally, the eventual probability of death will be the probability $|\beta|^2$ of atomic excitation; of course we are assuming that the cat will live forever if the decaying atom does not end its life. On the other hand, if $\alpha=1$ then the probability of (an eternal) cat remaining alive will approach 1, or more generally tend to $|\alpha|^2$.

Notice that even though we have begun to construct a complete compatible family of observations, the state of the total system is a quantum pure-state and given by the wave-function $\psi_t\in\mathcal{H}$ and it is a second-quantised superposition of cat states entangled with the atom.  However, recall that the  quantum filter is the conditional expectation $\epsilon_t:\mathcal{A}_t\rightarrow\mathcal{C}_t$. Now we shall see how the quantum filter may be given explicitly by the {filtering wave-function} $\Psi_t$ as
\begin{equation}
    \epsilon_t(X)=\Psi_t^\dag (X\otimes I)\Psi_t,
\end{equation}
where $\Psi_t$ lives in the Hilbert $\mathcal{C}_t$-module $\mathfrak{h}\otimes\mathcal{C}_t$. The filtering wave-function can be derived from a quantum Girsanov transformation
which ultimately gives us the output probability distribution $\hat\sigma$ for all of the different observations that could be made.

To this end we shall denote the input apparatus density by $\check\sigma=\Phi\Phi^\dag$ and obtain the following output apparatus density $\hat\sigma_t=\Xi^{}_t \Xi_t^{\dag}$ as
\begin{equation}
\psi_t^{}\psi_t^\dag=U_0^t(|\psi\rangle\langle\psi|\otimes\check\sigma)U_0^{t\dag}=G_t(|\psi\rangle\langle \psi|\otimes\hat\sigma_t)G_t^\dag=\Psi_t^{}\hat\sigma_t\Psi_t^\dag.
\end{equation}
This is the Girasnov transformation and basically involves changing $U_0^t\in\mathcal{B}(\mathcal{H})$ into $G_t\in\mathcal{A}_t$, where
$U_0^t(|\psi\rangle\otimes\Phi)=G_t(|\psi\rangle \otimes\Xi_t)$. The filtering wave-function is defined as $\Psi_t:=G_t(|\psi\rangle \otimes I)$ and is determined from the requirement that $\Psi_t^\dag\Psi_t^{}=P_t$. This means that $\Psi_t=\int|\vartheta\rangle\Psi_t(\vartheta)\langle\vartheta|\mathrm{d}\vartheta$ has the form
\begin{equation}\label{fwfwf}
\Psi_t(\vartheta)=\sum _{}c_{k_1\ldots k_n}^{-1}V_{k_n}\cdots V_{k_1}|\psi\rangle\otimes I(\vartheta_{[t})\otimes P_{k_n}(t_n)\otimes\cdots\otimes P_{k_1}(t_1)
\end{equation}
where the sum is take over all trajectories for which the normalisation factors $c_{k_1\ldots k_n}:=\|V_{k_n}\cdots V_{k_1}|\psi\rangle\|\neq 0$. In fact, $\Psi_t$ is a partial isometry, and the projector $E_t=\Psi_t^{}\Psi_t^\dag\in\mathcal{A}_t$ determines the filter as $E_t(X\otimes I)E_t=E_t\epsilon_t(X)$ and note that $E_t\Psi_t=\Psi_t$ too. Moreover, notice that the normalisation factors, and thus $\Psi_t$, depend on $|\psi\rangle$ (and also $\langle\psi|$ and this dependence can be non-linear), so it would be more appropriate to write $G_t(\psi)$, $\Xi_t(\psi)$ and so on, but for simplicity this notation has been dropped.

The output apparatus wave-function may be decomposed as $\Xi_t=\Phi_{[t}\otimes \Xi^t$, and is determined from the normalisation factors $c_{k_1\ldots k_n}$ as
\begin{equation}
\langle\vartheta|\Xi_t:=\Phi(\vartheta_{[t})\otimes e^{-\frac{1}{2}\nu t}\sum_{(k_1,\ldots,k_n)\in\Omega^n} c_{k_1\ldots k_n}\sqrt{\nu}|k_n\rangle\otimes\cdots\otimes \sqrt{\nu}|k_1\rangle,
\end{equation}
where $\vartheta^t=\{t_n>\cdots>t_1\}$ and  the $|\vartheta^t|=0$ case is just $\Phi(\vartheta_{[t})e^{-\frac{1}{2}\nu t}$, giving $\hat\sigma=\Xi^{}_t\Xi^{\dag}_t$. But note that, as a state on  $\mathcal{C}_t$, the output density $\hat\sigma$ is indistinguishable from the classical output distribution $\hat\varsigma _t=\check\sigma_{[t}\otimes \hat\varsigma^t$ where
\begin{equation}
    \hat\varsigma^t(\vartheta^t)=e^{-\nu t}\sum_{(k_1,\ldots,k_n)\in\Omega^n} |c_{k_1\ldots k_n}|^2 \nu P_{k_n}\otimes\cdots\otimes\nu P_{k_1}.
\end{equation}
This means that the expectation of any operator $A$ in $\mathcal{A}_t$ has the form
\begin{equation}
    \mathbb{E}_t(A):=\psi_t ^\dag A\psi^{}_t=\texttt{Tr}[\hat\varsigma_t \epsilon_t(A)].
\end{equation}

In order to understand this conditional expectation more intuitively let's reconsider it in the following way. Recall that each actual observation $\Pi_t(\vartheta)$, corresponding to a chain of individual observations at times $t_i\in\vartheta^t$, is a product of specific projections $P_{k_i}(t_i)$. And in turn this corresponds to a product of contractive operations $V_{k_i}$ on the quantum system. Thus for any such observation there is a corresponding (unnormalised) expectation of properties $X\in\mathcal{B}(\mathfrak{h})\subset\mathcal{A}_t$ given as
\begin{equation}\label{19}
\langle\psi|V_{k_n}^\dag\cdots V_{k_1}^\dag X V_{k_n}\cdots V_{k_1}|\psi\rangle = \psi_t^\dag (\vartheta)(X\otimes \Pi_t(\vartheta))\psi_t^{}(\vartheta)
\end{equation}
then we see that the conditional expectation is defined as
\begin{equation}
    \epsilon_t(X,\vartheta)=
    \sum^{i=\varnothing,0,1}_{\text{partitions of }\vartheta^t} \left(\frac{\psi_t^\dag (\vartheta)(X\otimes \Pi^i_t(\vartheta))\psi_t^{}(\vartheta)}{\psi_t^\dag (\vartheta)(1_\mathfrak{h}\otimes \Pi^i_t(\vartheta))\psi_t^{}(\vartheta)}\right)\Pi^i_t(\vartheta)
\end{equation}
and to see that this is the same as $\Psi_t^\dag(\vartheta)(X\otimes I)\Psi_t^{}(\vartheta)$ one need only consider the projections $\Pi^i_t(\vartheta)\Psi_t^\dag(\vartheta)(X\otimes I)\Psi_t(\vartheta)$ which equal $\Psi_t^\dag(\vartheta)(X\otimes \Pi^i_t(\vartheta))\Psi_t(\vartheta)$ and, in view of \eqref{fwfwf} and \eqref{19}, it is simply a matter of definition that
\begin{equation}
    \Psi_t^\dag(\vartheta)(X\otimes \Pi^i_t(\vartheta))\Psi_t(\vartheta)
    =\left(\frac{\psi_t^\dag (\vartheta)(X\otimes \Pi^i_t(\vartheta))\psi_t^{}(\vartheta)}{\psi_t^\dag (\vartheta)(1_\mathfrak{h}\otimes \Pi^i_t(\vartheta))\psi_t^{}(\vartheta)}\right)\Pi_t^i(\vartheta).
\end{equation}

With the interaction picture in mind we should actually be filtering the free Heisenberg operators $X(t)\in\mathcal{B}(\mathfrak{h})$. The dynamics of the quantum filtering of $\epsilon_t(X(t))$ may be determined from the dynamics of the
{filtering wave-function} $\Psi_t$ which is described by the  \emph{Belavkin equation} (if we worked in the Schr\"odinger picture this would reduce to the Schr\"odinger equation in the absence of observation) which, in the interaction picture for our Poisson-type cat observations, has the  form
\begin{equation}\label{fwf}
\mathrm{d}\Psi_t=L_k(\Psi_t)\Psi_t\mathrm{d}n^k_t,
\end{equation}
where $L_k(\Psi_t):=V_k(t) /|V_k(t)\Psi_t|-I$ with $|V_k\Psi_t|^2:=\epsilon_t(V_k^\dag V_k^{})$, and
$n^k_t$, $k=0,1,2$, are output counting processes (number operators) for this observation having evaluations
\begin{equation}
    n_t^k(\vartheta)=\sum_{t_i\in\vartheta^t} P_k(t_i)\otimes I(\vartheta\setminus t_i)
\end{equation}
and  intensities $\nu_k(t)=\nu\texttt{Tr}[\hat\varsigma_t\epsilon_t(V_k^\dag V_k^{})]=\nu\|(V_k\otimes I)\psi_t\|^2$ which may be calculated as $\nu_0=p\nu(1-|\beta|^2e^{-p\nu t})$, $\nu_1=p\nu|\beta|^2 e^{-p\nu t}$ and $\nu_2=q\nu$.
Then any freely evolving atomic property $X(t)$ may be projected onto the observer's algebra giving rise to the set of observable trajectories with each one carrying its own estimate of $X(t)$, and that is what the quantum filter $\epsilon_t(X(t))$ is. And its evolution is given by the \emph{classical} stochastic differential equation, represented on $\mathcal{F}$, as
\begin{equation}\label{feq}
\mathrm{d}\epsilon_t(X(t))=\epsilon_t(\mathrm{i}[H,X(t)])\mathrm{d}t+\kappa_{k,t}(X(t))\mathrm{d}n^k_t
\end{equation}
where $\kappa_{k,t}(X)=\epsilon_t(V_k^\dag X V_k^{})/\epsilon_t(V_k^\dag V_k^{})-\epsilon_t(X)$, and $k$ is summed over in \eqref{feq}.

From the explicit form of the conditional expectation for this model it becomes evident that $\Pi_t(\vartheta)\epsilon_t(X,\vartheta)=\langle g|X |g\rangle \Pi_t(\vartheta)$ for any trajectory $\Pi_t(\vartheta)$ if $\vartheta^t\neq \varnothing$, and if $\vartheta^t=\varnothing$ then $\epsilon_t(X,\vartheta)=\langle\psi|X|\psi\rangle \Pi^\varnothing_t(\vartheta)$. Consequently we see that
\begin{equation}
    \epsilon_t(X(t))=\langle\psi|X(t)|\psi\rangle P_t^\varnothing +\langle g|X(t) |g\rangle (P_t^0+P_t^1)
\end{equation}
from which we deduce that $\Psi_t=|\psi\rangle\otimes P_t^\varnothing +|g\rangle \otimes(P_t^0+P_t^1)$ in this  model. Further, the coefficients $\kappa_{k,t}(X(t))$ in \eqref{feq} may be calculated as
 \begin{equation}
  \kappa_{k,t}(X,\vartheta)= \big(\langle g|X| g\rangle-\langle\psi|X|\psi\rangle\big)I(\vartheta), \quad k=0,1
  \end{equation}
if $t=\min\{\vartheta\}$, otherwise
    \begin{equation}
\kappa_{0,t}(X,\vartheta)=0,\quad\textrm{and}\quad  \kappa_{1,t}(X,\vartheta)=-\epsilon_t(X,\vartheta),
  \end{equation}
with $\kappa_{2,t}(X,\vartheta)=0$ for all $\vartheta$.

Although the projectors $\Pi^i_t(\vartheta)$ represent the physical trajectories, i.e. a time-ordered sequence of observations determined by interaction with an atom, there are other observables of importance. Specifically, these are the number operators $n_t^k$, $k=0,1,2$, encountered above, counting the specific types of observation. In particular, the counting of the number of deaths of the cat is worth presenting. It corresponds to $n_t^1$ and its conditional expectation and expectation are
\begin{equation}
\epsilon_t(n_t^1)=P^{}_tn_t^1=P_t^1\quad\text{and}\quad \mathbb{E}_t[n_t^1]=|\beta|^2(1-e^{-\nu t})\leq 1,
\end{equation}
where the latter follows from the LTP property $\mathbb{E}_t[\epsilon_t(A)]=\mathbb{E}_t[A]$.

The fundamental role of these number operators in this this theory is that they are the random noises that are driving the dynamics \eqref{fwf} of the filtering wave-function $\Psi_t$. To make this clear: observation is a field of noise that drives the Lindblad dynamics of an open quantum system. This is non-trivial and in contrast to Zeno's paradox \cite{zeno}, which neglects the necessity of (quantum) stochastic calculus in modelling observation. Even in Lindblad's construction of quantum dynamics \cite{L} the Stinespring dilation of such was unknown until  the rigorous construction of quantum stochastic calculus in \cite{HP}.

\section{Concluding Remarks}
So, the complete quantum theory of a quantum object represented in a Hilbert space $\mathfrak{h}$ is given by a unitarily evolving pure-state $\mathbb{E}_t$ on a von Neumann algebra $\mathcal{A}_t$, represented on
a continuous tensor product Hilbert space $\mathcal{H}=\mathfrak{h}\otimes\mathcal{F}$, and defined as the commutant of the abelian von Neumann algebra $\mathcal{C}_t\subset\mathcal{B}(\mathcal{F})$ which forms a filtration $(\mathcal{C}_t)_{t\geq0}$.
The evolving pure-state is given by $\psi_t=U_0^t\psi$, and such evolution describes the interaction between the quantum object and an apparatus and resolves a quantum stochastic differential equation, generalising the usual Schr\"odinger equation.

The continuous tensor product structure is required in order to be able to continually condition, and thus predict, properties of the quantum object. The algebra  $\mathcal{A}_t$ is the algebra of all $\mathbb{E}_{t}|\mathcal{C}_{s}$-predictable quantum properties, $t\geq s$, and the quantum filtering $\epsilon_t(X)$ of $X\neg\mathcal{A}_t$ is evaluated as sums of compatible and orthogonal projections, each weighted with a specific expectation of $X$, resembling a sort-of estimated spectral decomposition of $X$ in $\mathcal{C}_t$.  These projections are  the trajectories of the apparatus, the observations, resulting from its coupling to the object.

The algebra $\mathcal{C}_t$ may be understood as the observer's history, or memory, and any quantum property affiliated with $\mathcal{B}(\mathfrak{h})$ forms no part of this history. Instead it may be interpreted as the observer's future, or perhaps the observer's present, either way it is distinct from the observer's past. Moreover, the filtering wave-function $\Psi_t$ serves as an interface between this non-observable quantum future and the observed classical past. In fact, $\Psi_t$ is the means by which quantum information is extracted and turned into trajectories. So perhaps it is $\Psi_t$ that should be interpreted as the actual act of observing, and then the trajectories are the result of this observation. In this way observation itself may be regarded as the creator of history.

 Let's now suppose  that the apparatus \emph{is} the observer, `you' are observing via your interaction with part of the universe. In a similar spirit to von Neumann  we shall decompose the world into three parts, but here these are:
 \begin{enumerate}[i)]
     \item The ever-present act of observing, or the experience of existing, if you will.
     \item The duration of this experience, our memory.
     \item Something that we believe to be here, giving rise to any experience we have, giving rise to our existence, and what we might identify with objective reality.
 \end{enumerate}
 This  third part is a result of our inference, it is something that we have deduced as a means to predict our future experiences. It also appears to be quantum mechanical, wave-like, in its structure, and inherently uncertain.  The second part, our memory, appears to be classical, particle-like trajectories, e.g. data being churned out in a laboratory, compatible projections of a quantum universe. As for the first part, the act of observing, this is  essentially the manner with which we interact with whatever it is that we are interacting with. Notice that this is quite different to the orthodox notion that objective reality corresponds to the classical  world.

 The basic quantum filtering model presented here would identify the inferred `objective reality' as  the quantum object $\mathcal{B}(\mathfrak{h})$ which is attributed to the source of our observation. Whilst the actual act of observation can ultimately be identified with the differential interaction operators $[L^\mu_\kappa(t)]=\boldsymbol{L}(t)$ generating the interaction $U_0^t$, of which the Hamiltonian is part. With this philosophy running in the background, what can be said about wave-function collapse?
 Well, all of the information about collapsing the wave-function $|\psi\rangle\in\mathfrak{h}$ is contained within the wave-function $\psi_t\in\mathcal{H}$ which is a superposition of all the different possible collapses, but it is by no means collapsed as it is just an isometry of $|\psi\rangle$. Moreover, all of the possible re-normalised collapsed  wave-functions are contained within the filtering wave-function $\Psi_t$. In particular, this is all encoded within the the quantum filter $\epsilon_t(X)=\Psi_t^\dag X\Psi_t^{}$, which is an \emph{observable} on $\mathcal{F}$. And $\Pi_t\epsilon_t(X)$ is the conditional expectation of $X$ for a given history (or memory) of observations $\Pi_t$.


One of the difficulties with building a model of the world in this way is that the interaction operators $L_\kappa^\mu$ have to be intuitively guessed. For example, in the interaction picture for the Schr\"odinger cat model we had only $L_k^i=S^i_k-\delta^i_k$, and these operators were constructed based on beliefs about what the quantum system under observation is and how it should behave in response to the different observations that we could make.
The usual approach to infer quantum structure comes from precision experiments. In the Stern-Gerlach experiment, for example, particles are beamed through a magnetic field and spin-structure can be inferred to explain the observations. However, again, this requires that the particles and the beam have a specific mathematical structure that we impose to explain the observation, and the manner in which the field and the beam interact is determined by our intuition (belief) about what we think should be happening in order to give rise to what we observe.
Meanwhile, spin-structure can also be derived theoretically, as done by Dirac, in order to form a more unified theoretical framework with the intention of having an inferred structure that can explain a bigger class of observations. Yet still the story is the same. We make observations and we form mathematically structured beliefs about why we made such observations.


As for objective reality: if it is understood that in experiments we are actually observing the  `apparatus' part of a composite system but nonetheless explaining our observation by virtue of the existence of a hidden quantum system represented by $\mathfrak{h}$, then, such a quantum system is inferred, but is it objectively real? Well, this isn't much of a question, because  we can \emph{define} objective reality as a subjective reality $\mathfrak{h}$ that is required in order to explain the  observations that we have i.e. something that (logically) \emph{must} be there in order to explain the observations we have. Then, with only our observations to lay claim to the reality of anything, we are none the wiser. But it would be nice to show that there is some kind of permanent structure to this underlying quantum world, something general and common to all of our inference resulting from our observations. In some sense this is what the Standard Model attempts to do, but such `fundamental particles' are still very specific to the kind of apparatus used, aren't they? Finding an underlying source that explains all observations of any kind (using any apparatus) is much harder, but if such a thing could be done then we could be so bold as to lay claim to an objective reality that is no more or less than the source of all our  observations. It may be that the only thing in common with all observations is the causal structure of the observation process itself, what we understand to be our experience of time. In this way quantum filtering may also me understood as a description of the emergence of time, as experienced by an observer. That is coming from two basic $\{\text{future},\text{past}\}$ degrees of freedom, or better still might be the terminology $\{\text{potential},\text{actual}\}$.

It is worth remarking upon that if we consider an abundance of observations of a particular kind then the `universal' quantum source presumably contains an abundance of individual quantum sources giving rise to such aforementioned observations. Again, this is the doctrine of the Standard Model and its fundamental `particles'. However, note that the concept of particles in Belavkin's filtering theory is something attributed to the apparatus, as this is what forms trajectories. Whereas such particles' quantum counterpart is regarded as wave-like, not particle-like. For example, the quantum electron is a wave, e.g. undergoing diffraction, but the observations of such an electron are particle-like, e.g. detections in a diffraction experiment or trajectories in a cloud chamber.

Finally, it might also seem reasonable that $\mathfrak{h}$ should not be different from $\mathcal{F}$, so that we can consider a more mutual structure to the world in which  two systems can observe one another. Then $\mathcal{H}=\mathcal{F}_1\otimes\mathcal{F}_2$, where $\mathfrak{h}=\mathcal{F}_1$ and $\mathcal{F}=\mathcal{F}_2$.
Note that the observation dynamics described throughout is referred to a time $t$ that is attributed to the observer $\mathcal{F}$, since all notions of past and future are given with respect to that observer. However, this need not be the time as experienced by the object $\mathfrak{h}$. So if the object is now also considered as an observing system then it will generally have it's own concept of time with respect to which it observes the other system.
Notice that system 1 can predict quantum properties in $\mathcal{B}(\mathcal{F}_2)$ and system 2 can predict  quantum properties in $\mathcal{B}(\mathcal{F}_1)$.
This structure was originally considered by Hudson, see \cite{DP} for example, and studied in the framework of Belavkin's QSC Formalism in \cite{MFB}, but the construction of a `universal' flow of time arising from two such mutually observing systems is yet to be found.

\end{document}